\documentclass[pdftex,USenglish,a4paper]{lipics}

\usepackage[sortcites=true, maxnames=99, firstinits=true, citestyle=alphabetic, style=numeric-comp, backend=bibtex]{biblatex}
\usepackage[capitalize]{cleveref}
\usepackage[inline]{enumitem}

\newlist{textenum}{enumerate*}{1}
\setlist[textenum,1]{label=(\roman*)}

\newcommand{\fNFA}{\mathsf{NFA}}
\newcommand{\fOCA}{\mathsf{OCA}}

\newcommand{\fCFG}{\mathsf{CFG}}

\newcommand{\fRBC}[1][OPT]{\ifthenelse{\equal{#1}{OPT}}{\mathsf{RBC}}{\mathsf{RBC}_{#1}}}

\mathchardef\mhyphen="2D
\newcommand{\coNP}{\mathsf{coNP}}
\newcommand{\NL}{\mathsf{NL}}
\newcommand{\coNL}{\mathsf{coNL}}
\newcommand{\cL}{\mathsf{L}}

\newcommand{\NP}{\mathsf{NP}}
\newcommand{\coNEXP}[1][OPT]{\ifthenelse{\equal{#1}{OPT}}{\mathsf{coNEXP}}{\mathsf{co}\mhyphen{}#1\mhyphen\mathsf{NEXP}}}
\newcommand{\Poly}{\mathsf{P}}
\newcommand{\Ideal}{\mathsf{Ideal}}
\newcommand{\PHPitwo}{\Pi_2^\Poly}
\newcommand{\PHSigmatwo}{\Sigma_2^\Poly}
\newcommand{\coNTIME}[1]{\mathsf{coNTIME}(#1)}

\newcommand{\Z}{\mathbb{Z}}
\newcommand{\N}{\mathbb{N}}
\newcommand{\Q}{\mathbb{Q}}
\newcommand{\cH}{\mathcal{H}}

\newcommand{\Powerset}[2][OPT]{\ifthenelse{\equal{#1}{OPT}}{\mathbb{P}(#2)}{\mathbb{P}_{#1}(#2)}}
\newcommand{\LinInd}[1]{\mathbb{I}(#1)}

\newcommand{\eword}{\varepsilon}

\newcommand{\A}{\mathcal{A}}
\newcommand{\B}{\mathcal{B}}
\newcommand{\G}{\mathcal{G}}
\newcommand{\T}{\mathcal{T}}

\newcommand{\langof}[1]{L(#1)}

\newcommand{\C}{\mathcal{C}}
\newcommand{\D}{\mathcal{D}}
\newcommand{\Dclosure}[1]{#1\mathord{\downarrow}}
\newcommand{\IdealUpper}[1]{|#1|_{\mathsf{I}}}
\newcommand{\Elen}[1]{|#1|_{\mathsf{E}}}
\newcommand{\PumpLang}[2]{P(#1,#2)}

\newcommand{\Word}[1]{w_{#1}}
\newcommand{\subword}{\preceq}
\newcommand{\IncProb}[2]{#1\subseteq_{\mathord{\downarrow}}#2}
\newcommand{\EqProb}[2]{#1 =_{\mathord{\downarrow}}#2}

\DeclareMathOperator{\spanop}{span}
\newcommand{\spanof}[1]{\spanop(#1)}

\newcommand{\citewithauthor}[1]{\citeauthor{#1}~\cite{#1}}

\newcommand{\yield}[1]{\mathsf{yield}(#1)}

\newcommand{\grammarstep}[1][YYY]{\ifthenelse{\equal{#1}{YYY}}{\Rightarrow}{\Rightarrow_{#1}}}
\newcommand{\grammarsteps}[1][YYY]{\ifthenelse{\equal{#1}{YYY}}{\Rightarrow^*}{\Rightarrow^*_{#1}}}
\newcommand{\grammarstepsn}[2][YYY]{\ifthenelse{\equal{#1}{YYY}}{\Rightarrow^{#2}}{\Rightarrow^{#2}_{#1}}}
\newcommand{\grammarstepp}[1][YYY]{\ifthenelse{\equal{#1}{YYY}}{\Rightarrow'}{\Rightarrow'_{#1}}}
\newcommand{\grammarstepsp}[1][YYY]{\ifthenelse{\equal{#1}{YYY}}{\Rightarrow'^*}{\Rightarrow'^*_{#1}}}
\newcommand{\grammarsteppp}[1][YYY]{\ifthenelse{\equal{#1}{YYY}}{\Rightarrow''}{\Rightarrow''_{#1}}}
\newcommand{\grammarstepspp}[1][YYY]{\ifthenelse{\equal{#1}{YYY}}{\Rightarrow''^*}{\Rightarrow''^*_{#1}}}
\newcommand{\grammarsteplin}[1][YYY]{\ifthenelse{\equal{#1}{YYY}}{\Rightarrow_{\mathsf{lin}}}{\Rightarrow_{#1,\mathsf{lin}}}}
\newcommand{\grammarstepslin}[1][YYY]{\ifthenelse{\equal{#1}{YYY}}{\Rightarrow_{\mathsf{lin}^*}}{\Rightarrow^*_{#1,\mathsf{lin}}}}

\newcommand{\LowerT}[1]{\Delta(#1)}

\newcommand{\mM}{\mathcal{M}}
\newcommand{\mN}{\mathcal{N}}

\title{The complexity of downward closure comparisons}
\author{Georg Zetzsche\footnote{This work is supported by a fellowship within the Postdoc-Program of the German Academic Exchange Service (DAAD).}}
\affil{LSV, CNRS \& ENS Cachan, Universit\'{e} Paris-Saclay, France \\ \texttt{zetzsche@lsv.fr}}
\authorrunning{G. Zetzsche}
\Copyright{Georg Zetzsche}

\theoremstyle{plain}
\newtheorem{thm}{Theorem}[section]
\newtheorem{prop}[thm]{Proposition}
\newtheorem{lem}[thm]{Lemma}
\newtheorem{cor}[thm]{Corollary}

\crefname{thm}{Theorem}{Theorems}
\Crefname{thm}{Theorem}{Theorems}
\crefname{prop}{Proposition}{Propositions}
\Crefname{prop}{Proposition}{Propositions}
\crefname{lem}{Lemma}{Lemmas}
\Crefname{lem}{Lemma}{Lemmas}
\crefname{cor}{Torollary}{Corollaries}
\Crefname{cor}{Corollary}{Corollaries}

\begin{document}
\maketitle
\begin{abstract}
The downward closure of a language is the set of all (not necessarily
contiguous) subwords of its members. It is well-known that the downward
closure of every language is regular. Moreover, recent results show
that downward closures are computable for quite powerful system models.

One advantage of abstracting a language by its downward closure is that then
equivalence and inclusion become decidable. In this work, we study the
complexity of these two problems. More precisely, we consider the following
decision problems: Given languages $K$ and $L$ from classes $\C$ and $\D$,
respectively, does the downward closure of $K$ include (equal) that of $L$?

These problems are investigated for finite automata, one-counter automata,
context-free grammars, and reversal-bounded counter automata. For each
combination, we prove a completeness result either for fixed or for arbitrary
alphabets.  Moreover, for Petri net languages, we show that both problems are
Ackermann-hard and for higher-order pushdown automata of order~$k$, we prove
hardness for complements of nondeterministic $k$-fold exponential time.
\end{abstract}
\section{Introduction}
The downward closure of a language is the set of (not necessarily contiguous)
subwords of its members. It is a well-known result
of~\citewithauthor{Haines1969} that the downward closure of \emph{every}
language is regular. Of course, it is not always possible to compute the
downward closure of a given language, but oftentimes it is. For example, it has
been shown to be computable for such powerful models as \emph{Petri net languages}
by~\citewithauthor{HabermehlMeyerWimmel2010} and \emph{higher-order pushdown automata}
by~\citewithauthor{HagueKochemsOng2016}. A sufficient condition for
computability can be found in~\cite{Zetzsche2015b}.

Moreover, not only are downward closures often computable, they are also a
meaningful abstraction of infinite-state systems. In a complex system, one can
abstract a component by the downward closure of the messages it sends to its
environment. This corresponds to the assumption that messages can be dropped on
the way. Furthermore, recent work
of~\citewithauthor{LaTorreMuschollWalukiewicz2015} shows that among other mild
conditions, computing downward closures is sufficient for verifying safety
conditions of parametrized asynchronous shared-memory systems.

The advantage of having an abstraction of an infinite-state systems as regular
languages is that the latter offer an abundance of methods for analysis.
An important example is deciding behavioral equivalence or inclusion. This is
notoriously hard to do and for nondeterministic infinite-state systems,
language equivalence and inclusion are usually undecidable. Using downward
closures, such behavioral comparisons can be made in an approximative manner.

Despite these facts, results about the complexity of deciding whether the
downward closure of one language includes or equals that of another mainly
considered regular languages. \citewithauthor{BachmeierLuttenbergerSchlund2015}
have shown that the equivalence problem for downward closures of two given NFAs
is $\coNP$-complete.  \citewithauthor{KarandikarNiewerthSchnoebelen2016}
strengthened $\coNP$-hardness to the case of DFAs over binary alphabets and
proved $\coNP$-completeness for the inclusion variant. They also obtained
$\NL$-completeness of inclusion in the case of NFAs over a unary alphabet.
Together with exponential-time downward closure
constructions~\cite{vanLeeuwen1978,Courcelle1991,Okhotin2010,GruberHolzerKutrib2009,BachmeierLuttenbergerSchlund2015},
these results imply that equivalence and inclusion are in $\coNEXP$ for
context-free grammars. \citewithauthor{RampersadShallitXu2012} proved that one
can decide in linear time whether the downward closure of a given NFA contains
all words. Subsequently, \citewithauthor{KarandikarNiewerthSchnoebelen2016}
showed that this problem is $\NL$-complete. Similar questions have been studied 
for upward
closures~\cite{BachmeierLuttenbergerSchlund2015,KarandikarNiewerthSchnoebelen2016}.

Previous work on downward closures of infinite-state systems has mainly focused
on mere
computability~\cite{Courcelle1991,vanLeeuwen1978,HabermehlMeyerWimmel2010,HagueKochemsOng2016,
Zetzsche2015b,Zetzsche2015a,Abdulla2004,AbdullaBoassonBouajjani2001} or on
descriptional
complexity~\cite{AtChHoKuSaZe2015,GruberHolzerKutrib2007,GruberHolzerKutrib2009,Okhotin2010,KarandikarNiewerthSchnoebelen2016}.
This work studies the complexity of the inclusion and the equivalence problem
of downward closures between some prominent types of system models---finite
automata, one-counter automata, reversal-bounded counter
automata~\cite{Ibarra1978}, and context-free grammars. More precisely, we are
interested in the following questions: For two system models $\mM$ and $\mN$
and languages $L$ and $K$ generated by some device in $\mM$ and $\mN$,
respectively, what is the complexity of \begin{textenum}\item deciding whether
$\Dclosure{K}\subseteq \Dclosure{L}$ (\emph{downward closure inclusion
problem}) or \item deciding whether $\Dclosure{K}=\Dclosure{L}$ (\emph{downward
closure equivalence problem})?\end{textenum}

\subparagraph{Contribution} We determine the complexity of the downward closure
inclusion problem and the downward closure equivalence problem among finite automata,
one-counter automata, reversal-bounded counter automata (either with a fixed
number of counters and reversals or without), and context-free grammars. 

For the inclusion problem, we prove completeness results in all cases except
for two.  The complexities range from $\coNP$ over $\PHPitwo$ to $\coNEXP$ (see
\cref{tableresults}). The two cases for which we provide no completeness
compare context-free grammars or general reversal-bounded counter automata on
the one side with reversal-bounded counter automata with a fixed number of
counters and reversals on the other side.  However, we prove that both of these
problems are $\coNP$-complete for each fixed input alphabet.
For the equivalence problem, the situation is similar. We prove completeness for
each of the cases except for the combination above. Again, fixing the
alphabet leads to $\coNP$-completeness.

The tools developed to achieve these results fall into three categories.
First, there are several generic results guaranteeing small witnesses to yield
upper bounds.   Second, we prove model-specific results about downward closures
that yield the upper bounds in each case.  Third, we have a general method to
prove lower bounds for downward closure comparisons.  In fact, it applies to
more models than the above: We prove that for Petri net languages, the two
comparison problems are Ackermann-hard. For higher-order pushdown automata of
order~$k$, we show $\coNEXP[k]$-hardness. 

\subparagraph{Related work} Another abstraction of formal languages is the
well-known Parikh image~\cite{Parikh1966}. The Parikh image of a language
$L\subseteq X^*$ contains for each word $w\in L$ a vector in $\N^{|X|}$ that
counts the number of occurrences of each letter. For some language classes, it
is known that their Parikh image is effectively semilinear, which implies
decidability of the inclusion and equivalence problem for Parikh images. The
investigation of these problems' complexity has been initiated
by~\citewithauthor{Huynh1985} in 1985, who showed that this problem is
$\PHPitwo$-hard and in $\coNEXP$ for regular and context-free languages.
\citeauthor{KopczynskiTo2010}~\cite{Kopczynski2015,KopczynskiTo2010} have then
shown that these problems are $\PHPitwo$-complete for fixed alphabets.  Only
very recently, \citewithauthor{HaaseHofman2016} have shown that the case of
general alphabets is $\coNEXP$-complete.

\section{Concepts and Results}\label{results}
\newcommand{\BLS}{\cite{BachmeierLuttenbergerSchlund2015,KarandikarNiewerthSchnoebelen2016}}
\newcommand{\BLSA}{\cite{BachmeierLuttenbergerSchlund2015,KarandikarNiewerthSchnoebelen2016,AtChHoKuSaZe2015}}
\begin{table}[t]
\begin{center}
\begin{tabular}{|c|llllll|}\hline
                        & $\Ideal$      & $\fNFA$       & $\fOCA$       & $\fRBC[k,r]$  	& $\fCFG$       & $\fRBC$       \\ \hline
$\Ideal$                & $\in\cL$      & $\NL$         & $\NL$         & $\NL$         	& $\Poly$       & $\NP$         \\
$\fNFA$                 & $\NL$         & $\coNP$~\BLS  & $\coNP$~\BLSA & $\coNP$       	& $\coNP$       & $\PHPitwo$    \\
$\fOCA$                 & $\NL$         & $\coNP$~\BLSA & $\coNP$~\BLSA & $\coNP$       	& $\coNP$       & $\PHPitwo$    \\
$\fRBC[k,r]$            & $\NL$         & $\coNP$       & $\coNP$       & $\coNP$       	& $\coNP$       & $\PHPitwo$    \\
$\fCFG$                 & $\Poly$       & $\coNP$       & $\coNP$       & $\coNP^\dagger$       & $\coNEXP$     & $\coNEXP$     \\
$\fRBC$                 & $\coNP$       & $\coNP$       & $\coNP$       & $\coNP^\dagger$       & $\coNEXP$     & $\coNEXP$     \\ \hline
\end{tabular}
\end{center}
\caption{Complexity of the inclusion problem. The entry in row $\mM$ and column
$\mN$ is the complexity of $\IncProb{\mM}{\mN}$. Except in the case
$\IncProb{\Ideal}{\Ideal}$, all entries indicate completeness.  A $\dagger$
means that the entry refers to the fixed alphabet case (for at least two letters).}
\label{tableresults}
\end{table}

If $X$ is an alphabet, $X^*$ ($X^{\le n}$) denotes the set of all words (of
length $\le n$) over $X$.  The empty word is denoted by $\eword\in X^*$.  For
words $u,v\in X^*$, we write $u\preceq v$ if $u=u_1\cdots u_n$ and
$v=v_0u_1v_1\cdots u_nv_n$ for some $u_1,\ldots,u_n,v_0,\ldots,v_n\in X^*$.  It
is well-known that $\preceq$ is a well-quasi-order on $X^*$ and that therefore
the \emph{downward closure} $\Dclosure{L}=\{u\in X^* \mid \exists v\in L\colon
u\preceq v\}$ is regular for every $L\subseteq X^*$~\cite{Haines1969}. An
\emph{ideal} is a set of the form $Y_0^*\{x_1,\eword\}Y_1^*\cdots
\{x_n,\eword\}Y_n^*$, where $Y_0,\ldots,Y_n$ are alphabets and $x_1,\ldots,x_n$
are letters.  We will make heavy use of the fact that every downward closed
language can be written as a finite union of ideals, which was first discovered
by~\citewithauthor{Jullien1969}. By $\Powerset{S}$, we denote the powerset of the set $S$.

A \emph{finite automaton} is a tuple $\A=(Q,X,\Delta,q_0,Q_f)$, where $Q$ is a
finite set of \emph{states}, $X$ is its input alphabet, $\Delta\subseteq
Q\times X^*\times Q$ is a finite set of \emph{edges}, $q_0\in Q$ is its
\emph{initial state}, and $Q_f\subseteq Q$ is the set of its \emph{final
states}. The language accepted by $\A$ is denoted $\langof{\A}$.  Sometimes, we
write $|\A|$ for the number of states of $\A$.

A \emph{context-free grammar} is a tuple
$\G=(N,T,P,S)$ where $N$ and $T$ are pairwise disjoint alphabets, whose members are called the
\emph{nonterminals} and \emph{terminals}, respectively. $S\in N$ is the
\emph{start symbol} and $P$ is the finite set of \emph{productions} of the form
$A\to w$ with $A\in N$ and $w\in T^*$.  The language generated by $\G$ is
defined as usual.

\subparagraph{One-counter Automata} A \emph{one-counter automaton (OCA)} is a
nondeterministic finite automaton that has access to one counter that assumes
natural numbers as values. The possible operations are \emph{increment},
\emph{decrement}, and \emph{test for zero}. We will not require a formal
definition, since in fact, all we need is the well-known fact that membership
and emptiness are $\NL$-complete and the recent result that given an OCA $\A$,
one can compute in polynomial time an NFA $\B$ with
$\langof{\B}=\Dclosure{\langof{\A}}$~\cite{AtChHoKuSaZe2015}.

\subparagraph{Reversal-bounded counter automata} Intuitively, an
\emph{$r$-reversal-bounded $k$-counter automaton}~\cite{Ibarra1978} (short
$(k,r)$-RBCA) is a nondeterministic finite automaton with $k$ counters that can
store natural numbers. For each counter, it has operations \emph{increment},
\emph{decrement}, and \emph{zero test}.  Moreover, a computation is only valid
if each counter \emph{reverses} at most $r$ times. Here, a computation
\emph{reverses} a counter $c$ if on $c$, it first executes a sequence of
increments and then a decrement command or vice versa. See \cite{Ibarra1978}
for details.

Instead of working directly with RBCA, we will work here with the model of
\emph{blind counter automata}~\cite{Greibach1978}. It is not as well-known as
RBCA, but simpler and directly amenable to linear algebraic methods.  A
\emph{blind $k$-counter automaton} is a tuple $\A=(Q,X,q_0,\Delta,Q_f)$, where $Q$,
$X$, $q_0$, and $Q_f$ are defined as in NFAs, but $\Delta$ is a finite subset of $Q\times
(X\cup \{\eword\}) \times \{-1,0,1\}^k\times Q$.  A \emph{walk} is a word
$\delta_1\cdots\delta_m\in\Delta^*$ where $\delta_i=(p_i, x_i, d_i, p'_i)$ for
$i\in[1,m]$ and $p'_j=p_{j+1}$ for $j\in[1,m-1]$. The \emph{effect} of the walk
is $d_1+\cdots+d_m$.  Its \emph{input} is $x_1\cdots x_m\in X^*$.  If the walk
has effect $0$ and $p_0=q_0$ and $p_m\in Q_f$, then the walk is
\emph{accepting}.  The \emph{language accepted by $\A$} is the set of all
inputs of accepting walks.

Using blind counter automata is justified because to each $(k,r)$-RBCA, one can
construct in logarithmic space a language-equivalent
$(kr,1)$-RBCA~\cite{BakerBook1974}, which is essentially a blind $kr$-counter
automaton. On the other hand, every blind $k$-counter automaton can be turned
in logarithmic space into a $(k+1,1)$-RBCA~\cite{JantzenKurganskyy2003}.
Hence, decision problems about $(k,r)$-RBCA for fixed $k$ and $r$ correspond to
problems about blind $k$-counter automata for fixed $k$.

In the following, by a \emph{model}, we mean a way of specifying a language.
In order to succinctly refer to the different decision problems, we use symbols
for the models above. By $\Ideal$, $\fNFA$, $\fOCA$, $\fRBC[k,r]$, $\fRBC$, $\fCFG$, we
mean ideals, finite automata, OCA, RBCA with a fixed number of counters and
reversals, general RBCA, and context-free grammars, respectively. Then, for
$\mM,\mN\in\{\Ideal,\fNFA,\fOCA, \fRBC[k,r],\fRBC,\fCFG\}$, we consider the following
problems. In the \emph{downward closure inclusion problem}
$\IncProb{\mM}{\mN}$, we are given a language $K$ in $\mM$ and a language $L$ in $\mN$
and are asked whether $\Dclosure{K}\subseteq\Dclosure{L}$.
For the \emph{downward closure equivalence problem} $\EqProb{\mM}{\mN}$, the input is the same,
but we are asked whether $\Dclosure{K}=\Dclosure{L}$.

\subparagraph{Results} The complexity results for the inclusion problem are
summarized in \cref{tableresults}. For the equivalence problem, we will see
that every hardness result for $\IncProb{\mM}{\mN}$ also holds for
$\EqProb{\mM}{\mN}$.  Since for non-ideal models, the appearing complexity
classes are pairwise comparable, this implies that the complexity for
$\EqProb{\mM}{\mN}$ is then the harder of the two classes for
$\IncProb{\mM}{\mN}$ and $\IncProb{\mN}{\mM}$. For example, the problem
$\EqProb{\fNFA}{\fRBC}$ is $\PHPitwo$-complete and for fixed alphabets,
$\EqProb{\fRBC[k,r]}{\fCFG}$ is $\coNP$-complete.

\section{Ideals and Witnesses}\label{witnesses}
Our algorithms for inclusion use three types of witnesses.  The first type is a
slight variation of a result of \cite{BachmeierLuttenbergerSchlund2015}.  The
latter authors were interested in equivalence problems, which caused their
bound to depend on both input languages. The proof is essentially the same.
\begin{prop}[Short witness]\label{shortwitness}
If $\A$ is an NFA and $\Dclosure{K}\not\subseteq\Dclosure{\langof{\A}}$, then
there exists a $w\in\Dclosure{K}\setminus\Dclosure{\langof{\A}}$ with $|w|\le
|\A|+1$.
\end{prop}

The other types of witnesses strongly rely on ideals, which requires some
notation. An ideal is a product $I=Y_0^* \{x_1,\eword\} Y_1^* \cdots
\{x_n,\eword\} Y_n^*$ where the $Y_i$ are alphabets and the $x_i$ are letters.
Its \emph{length} $\IdealUpper{I}$ is the smallest $n$ such that $I$ can be
written in this form. Since every downward closed language can be written as a
finite union of ideals, we can extend this definition to languages:
$\IdealUpper{L}$ is the smallest $n$ such that $\Dclosure{L}$ is a union of
ideals of length $\le n$. 

Sometimes, it will be convenient to work with a different length measure of
ideals.  An \emph{ideal expression (of length $n$)} is a product $L_1\cdots
L_n$, where each $L_i$ is of the form $Y^*$ or $\{x,\eword\}$, where $Y$ is an
alphabet and $x$ is a letter. Note that $Y^*=Y^*\{x,\eword\}$ if $x\in Y$ and
$\{x,\eword\}=\emptyset^*\{x,\eword\}$.  Therefore, an ideal expression of
length $n$ defines an ideal of length $\le n$. In analogy to
$\IdealUpper{\cdot}$, for a language $L$, we define its \emph{expression
length} $\Elen{L}$ to be the smallest $n$ such that $\Dclosure{L}$ can be
written as a finite union of ideal expressions of length $\le n$.  The expression
length has the advantage of being subadditive: For languages $K,L$ we have
$\Elen{KL}\le \Elen{K}+\Elen{L}$. Moreover, we have $\IdealUpper{L}\le
\Elen{L}\le 2\IdealUpper{L}+1$.

The measure $\IdealUpper{\cdot}$ turns out to be instrumental for
the inclusion problem. Note that $\Dclosure{K}\not\subseteq\Dclosure{L}$ if and
only if there is an ideal $I\subseteq \Dclosure{K}$ of length
$\le\IdealUpper{K}$ with $I\not\subseteq\Dclosure{L}$.  We can therefore guess
ideals and check inclusion for them. From now on, we assume alphabets to come
linearly ordered. This means for every alphabet $Y$, there is a canonical word
$w_Y$ in which every letter from $Y$ occurs exactly once.
\begin{prop}[Ideal witness]\label{idealwitness}
Let $I=Y_0^* \{x_1,\eword\} Y_1^* \cdots \{x_n,\eword\} Y_n^*$. Then the following are equivalent:
\begin{textenum}
\item\label{iwit:ideal} $I\subseteq \Dclosure{L}$.
\item\label{iwit:every} $\Word{Y_0}^m x_1 \Word{Y_1}^m \cdots x_n \Word{Y_n}^m\in \Dclosure{L}$ for every $m\ge\IdealUpper{L}+1$.
\item\label{iwit:some} $\Word{Y_0}^m x_1 \Word{Y_1}^m \cdots x_n \Word{Y_n}^m\in \Dclosure{L}$ for some $m\ge \IdealUpper{L}+1$.
\end{textenum}
\end{prop}
A word of the form $\Word{Y_0}^m x_1 \Word{Y_1}^m \cdots x_n \Word{Y_n}^m\in
\Dclosure{L}$ with $m\ge \IdealUpper{L}+1$ is therefore called an \emph{ideal
witness for $I$ and $L$}. The proof of \cref{idealwitness} is a simple pumping
argument based on the fact that an ideal of length $\le m$ admits an NFA with
$\le m+1$ states.  Ideal witnesses are useful when we have a small bound on
$\IdealUpper{K}$ and $\IdealUpper{L}$ but only a large bound on the NFA size of
$\Dclosure{L}$. Observe that putting a bound on $\IdealUpper{L}$ amounts to
proving a pumping lemma: We have $\IdealUpper{L}\le n$ if and only if
for every $w\in L$, there is an ideal $I$ with $\IdealUpper{I}\le n$ and $x\in
I\subseteq\Dclosure{L}$. 

However even if, say,  $\IdealUpper{K}$ is polynomial and
$\IdealUpper{L}$ is exponential, ideal witnesses can be stored succinctly in
polynomial space, by keeping a binary representation of the power $m$. For
instance, this will be used in the case $\IncProb{\fNFA}{\fRBC}$.

Sometimes, we have a small bound on $\IdealUpper{L}$, but $\IdealUpper{K}$ may
be large.  Then, ideal witnesses are too large to achieve an optimal algorithm.
In these situations, we can guarantee smaller witnesses if we fix
the alphabet.
\begin{prop}[Small alphabet witness]\label{alphabetwitness}
Let $K,L\subseteq X^*$.
If $\Dclosure{K}\not\subseteq\Dclosure{L}$, then there exists a $w\in
\Dclosure{K}\setminus \Dclosure{L}$ with $|w|\le |X|\cdot (\IdealUpper{L}+1)^{|X|}$.
\end{prop}
The proof of \cref{alphabetwitness} is more involved than
\cref{idealwitness,shortwitness}. Note that a naive bound can be obtained by
intersecting exponentially (in $\IdealUpper{L}$) many automata for the ideals
of $\Dclosure{L}$ and complementing the result.  This would yield a doubly
exponential (in $\IdealUpper{L}$) bound, even considering the fact that ideals
have linear-sized DFAs. We can, however, use the latter fact in a different way.

A DFA is \emph{ordered} if its states can be partially ordered so that for
every transition $p\xrightarrow{x}q$, we have $p\le q$. In other words, the
automaton is acyclic except for loop transitions. The following
\lcnamecref{idealdfa} is easy to see: In order to check membership in an ideal,
one just has to keep a pointer into the expression that never moves left.
\begin{lem}\label{idealdfa}
Given an ideal representation of length $n$, one can construct in logarithmic
space an equivalent ordered DFA with $n+2$ states.
\end{lem}

An ordered DFA \emph{cycles} at a position of an input word if that
position is read using a loop.  The following \lcnamecref{ordereddfacycles} is
the key idea behind \cref{alphabetwitness}. Together with \cref{idealdfa},
it clearly implies \cref{alphabetwitness}. For unary alphabets, it is easy to
see.  We use induction on $|X|$ and show, roughly speaking, that without such a
position, no strict subalphabet can be used for too long. Then, all letters
have to appear often, meaning a state has to repeat after seeing the whole
alphabet. Hence, the automaton stays in this state until the end.
\begin{lem}\label{ordereddfacycles}
If $w\in X^*$ with $|w|>|X|\cdot (n-1)^{|X|}$, then $w$ has a position
at which every ordered $n$-state DFA cycles.
\end{lem}

\section{Insertion trees}
In \cref{counterautomata}, we will show upper bounds for the size of downward
closure NFAs and for ideal lengths for counter automata. These results employ
certain decompositions of NFA runs into trees, which we discuss here.

Let $\A=(Q,X,\Delta,q_0,Q_f)$ be a finite automaton.  A \emph{walk} is
a word $w=\delta_1\cdots\delta_m\in\Delta^*$ where $\delta_i=(p_i, x_i, p'_i)$
for $i\in[1,m]$ and $p'_j=p_{j+1}$ for $j\in[1,m-1]$. The walk is a
\emph{($p_1$-)cycle} if $p_1=p'_m$.  In this case, we define $\sigma(w):=p_1$.
A cycle is \emph{prime} if $p_i=p_1$ implies $i=1$.  A cycle is \emph{simple}
if $p_i=p_j$ implies $i=j$. A state $q$ \emph{occurs} on the cycle if $p_i=q$
for some $i$. If $i\ne 1$, then $q$ occurs \emph{properly}.

A common operation in automata theory is to take a run and delete cycles until
the run has length at most $|Q|$. The idea behind an insertion tree is to
record where we deleted which cycles. This naturally leads to a tree.

For our purposes, trees are finite, unranked and ordered.
An \emph{insertion tree}
is a tree $t=(V,E)$ together with a map $\gamma\colon V\to\Delta^*$ that assigns to each vertex
$v\in V$ a simple cycle $\gamma(v)$ such that
if $u$ is the parent of $v$, then $\sigma(\gamma(v))$ properly occurs in $\gamma(u)$.
Note that we allow multiple children for a state that occurs in $\gamma(u)$.

Since $t$ is ordered and in every simple cycle there is at most one proper
occurrence of each state, an insertion tree defines a unique (typically not
simple) cycle $\alpha(t)$. Formally, if $t$ is a single vertex $v$, then
$\alpha(t):=\gamma(v)$. If $t$ consists of a root $r$ and subtrees
$t_1,\ldots,t_s$, then $\alpha(t)$ is obtained by inserting each $\alpha(t_i)$
in $\gamma(r)$ at the (unique) occurrence of $\sigma(\alpha(t_i))$. The
\emph{height} of an insertion tree is the height of its tree. 
\begin{lem}\label{trees}
Every prime cycle of $\A$ admits an insertion tree of height at most $|Q|$.
\end{lem}
The idea is to pick a cycle $c$ strictly contained in the prime cycle, but of
maximal length. Then, after removing $c$, no state occurs both before and after
the old position of $c$. This forces any insertion tree $t$ of the
remainder to place this position in the root. We then apply induction to the
subtrees of $t$ and to $c$. The resulting trees can then all be attached to the
root, increasing the height by at most one.

One application of \cref{trees} is to construct short ideals in a pumping lemma
for counter automata.  Part of this construction is independent from counters, so
we stay with NFAs for a moment.  Suppose we have an insertion tree $t=(V,E)$
with map $\gamma\colon V\to\Delta^*$ and a subset $F\subseteq V$, whose members
we call \emph{fixed vertices} or \emph{fixed cycles}. Those in $V\setminus F$
are called \emph{pumpable vertices/cycles}.  

We use fixed and pumpable vertices to guide a pumping process as follows. A
sequence $s=t_1\cdots t_m$ of insertion trees is called \emph{compatible} if
$\sigma(\alpha(t_1))=\cdots=\sigma(\alpha(t_m))$.  We assume that we have a
global set $F$ of vertices that designates the fixed vertices for all these
trees.  Suppose $v$ is a pumpable vertex. We obtain new compatible
sequences in two ways: 
\begin{itemize}
\item  Let  $v_1,\ldots,v_\ell$ be the children of $v$.  We choose
$i\in[0,\ell]$ and split up $v$ at $i$, meaning that we create a new vertex
$v'$ with $\gamma(v')=\gamma(v)$ to the right of $v$ and move
$v_{i+1},\ldots,v_\ell$ (and, of course, their subtrees) to $v'$.
\item If the whole subtree under $v$ is pumpable (we call such subtrees
\emph{pumpable}), then we can duplicate this subtree and attach its root
somewhere as a sibling of $v$. 
\end{itemize}
If $v$ is a root, these operations mean that we introduce a new tree in the
sequence.  If a compatible sequence $s'$ is obtained from $s$ by repeatedly
performing these operations, we say that $s'$ is obtained by \emph{pumping
$s$}. This allows us to define the following language:
\[ \PumpLang{t_1\cdots t_m}{F}=\{\iota(\alpha(t'_1)\cdots \alpha(t'_k)) \mid \text{$t'_1\cdots t'_k$ results from pumping $t_1\cdots t_m$} \}. \]
Here, for a walk $w$, $\iota(w)$ denotes the input word read by $w$.  The
following \lcnamecref{pumpideal} will yield the desired short ideals.
\begin{lem}\label{pumpideal}
Let $s=t_1\cdots t_m$ be a compatible sequence of insertion trees of height
$\le h$ and let $F$ be a set of fixed vertices. Then, the language
$\Dclosure{\PumpLang{s}{F}}$ is an ideal that satisfies
$\Elen{\Dclosure{\PumpLang{s}{F}}}\le h|F|(2|Q|+|F|)^2$.
\end{lem}
Roughly speaking, the pumping process is designed so that pumpable subtrees
only cause alphabets $Y$ in factors $Y^*$ of the ideal to grow and thus do not 
affect the ideal length. Hence, the only vertices that contribute to the length
are those that are ancestors of vertices in $F$. Since the trees have height
$\le h$, there are at most $h|F|$ such ancestors.

\section{Counter Automata}\label{counterautomata}
In this section, we construct downward closure NFAs for counter automata and
prove upper bounds for ideal lengths. Mere computability of downward closures
of blind counter automata can be deduced from computability for Petri net
languages~\cite{HabermehlMeyerWimmel2010}.  However, that necessarily results
in non-primitive recursive automata (see \cref{hardness}). As a special case of
stacked counter automata, blind counter automata were provided with a new
construction method in~\cite{Zetzsche2015a}. That algorithm, however, yields
automata of non-elementary size. Here, we prove an exponential bound.
\begin{thm}\label{dc:blind}
For each $n$-state blind $k$-counter automaton $\A$, there is an 
NFA $\B$ with $\langof{\B}=\Dclosure{\langof{\A}}$ and $|\B|\le
(3n)^{5nk+7k^3}$.  Moreover, $\B$ can be computed in exponential time.
\end{thm}

\subparagraph{Linear Diophantine equations} In order to show correctness of our
construction, we employ a result of \citewithauthor{Pottier1991}, which bounds
the norm of minimal non-negative solutions to a linear Diophantine equation.
Let $A\in\Z^{k\times m}$ be an integer matrix. We write $\|A\|_{1,\infty}$
for $\sup_{i\in[1,k]}(\sum_{j\in[1,m]} |a_{ij}|)$, where $a_{ij}$ is the
entry of $A$ at row $i$ and column $j$. A solution $x\in\N^m$ to the equation
$Ax=0$ is \emph{minimal} if there is no $y\in\N^m$ with $Ay=0$ and $y\le x$,
$y\ne x$. The set of all solutions clearly forms a submonoid of $\N^m$, which
is denoted $M$.  The set of minimal solutions is denoted $\cH(M)$ and called
the \emph{Hilbert basis} of $M$. Let $r$ be the rank of $A$. Pottier showed the
following.  
\begin{thm}[\citewithauthor{Pottier1991}]\label{pottier}
For each $x\in \cH(M)$, $\|x\|_1\le (1+\|A\|_{1,\infty})^r$.
\end{thm}

By applying \cref{pottier} to the matrix $(A|-b)$, it is easy to
deduce that for each $x\in\N^m$ with $Ax=b$, there is a $y\in\N^m$ with $Ay=b$,
$y\le x$, and $\|y\|_1\le (1+\|(A|-b)\|_{1,\infty})^{r+1}$. 

\subparagraph{Automata for the downward closure}
Let $\A$ be a blind $k$-counter automaton with $n$ states.
The idea of the construction of $\B$ is to traverse insertion trees of
prime cycles of $\A$. Although insertion trees were introduced for finite automata,
they also apply to blind counter automata if we regard the counter updates as input symbols.
$\B$ keeps track of where it is in the tree using a
stack of bounded height. The stack alphabet will be $\Gamma=Q \times [-n,n]^k$.
We define $B=n+n\cdot (3n)^{(k+1)^2}$. The state set of our automaton $\B_1$ is the following:
\[ Q_1 = Q \times \Gamma^{\le n}\times [-B,B]^k\times \Powerset{[-n,n]^k}\times \Powerset{[-n,n]^k}. \]
Here, the number of states is clearly doubly exponential, but we shall make the
automaton smaller in two later steps. 
The idea behind $\B_1$ is that counter values in the interval $[-B,B]$ are
simulated precisely (in the factor $[-B,B]^k$). Roughly speaking, whenever we
encounter a cycle, we can decide whether to \begin{textenum}\item add its
effect to this precise counter or to \item remember the effect as ``must be
added at least once''.\end{textenum} We call the former \emph{precise
cycles}; the latter are dubbed \emph{obligation cycles} and are stored in
the first factor $\Powerset{[-n,n]^k}$.  In either case, the effect of a cycle
is kept as ``repeatable''  in the second factor $\Powerset{[-n,n]^k}$.  

In order to be able to guess for each cycle whether it should be a precise
cycle or an obligation cycle, we traverse an insertion tree of (the prime
cycles on) a walk of $\A$.  On the stack (the factor $\Gamma^{\le n}$), we keep
the cycles that we have started to traverse.  Suppose we are 
executing a cycle in a vertex $v$ and the path from the root to $v$ consists of
the vertices $v_1,\ldots,v_m$. Let $\gamma(v_i)$ be a $q_i$-cycle for
$i\in[1,m]$.  Then, the stack content is $(q_1, u_1)\cdots (q_m,u_m)$, where
$u_i$ is the effect of the part of $\gamma(v_i)$ that has already been
traversed.

In the end, we verify that \begin{textenum} \item \label{counter:acceptance:precise} the precise counter is zero
and \item \label{counter:acceptance:cancellable} one can add up obligation cycles (each of them at least once) and
repeatable cycles to zero.\end{textenum} The latter condition is captured in
the following notion.  Let $S,T\subseteq\Z^k$ be finite sets with
$S=\{u_1,\ldots,u_s\}$, $T=\{v_1,\ldots,v_t\}$. We call the pair $(S,T)$
\emph{cancellable} if there are $x_1,\ldots,x_s\in\N\setminus\{0\}$ and
$y_1,\ldots,y_t\in\N$ with $\sum_{i=1}^s x_iu_i+\sum_{i=1}^t y_iv_i=0$.  In
particular, $(\emptyset,T)$ is cancellable for any finite $T\subseteq\Z^k$.
Together, \labelcref{counter:acceptance:precise} and \labelcref{counter:acceptance:cancellable}
guarantee that the accepted word is in the downward closure: They imply
that we could have executed all of the obligation cycles and some others
(again) to fulfill our obligation. Hence, there is a run of $\A$ accepting a
superword.

The number of cycles we can use as precise cycles is limited by
the capacity $B$ of our precise counter. We shall apply \cref{pottier} to show
that there is always a choice of cycles to use as precise cycles so as to reach
zero in the end and not exceed the capacity.

The first type of transition in $\B_1$ is the following.  For each transition
$(p,a,d,q)\in \Delta$ and state $(p,\eword,v,S,T)\in Q_1$ such that $v+d\in
[-B,B]^k$, we have a transition
\begin{equation} (p,\eword,v,S,T)\xrightarrow{a} (q,\eword,v+d,S,T). \label{edges:b1:ord} \end{equation}
These allow us to simulate transitions in a walk of $\A$ that are not part of a cycle.
We can guess that a cycle is starting. If we are in state $p$, then we push
$(p,0)$ onto the stack:
\begin{equation} (p, w, v, S, T)       \xrightarrow{\eword}     (p, w(p, 0), v, S, T).\label{edges:b1:push} \end{equation}
While we are traversing a cycle, new counter effects are stored in the topmost
stack entry. For each transition $(p,a,d,q)\in\Delta$ and state $(p, w(r, u), v, S,
T)\in Q_1$ such that $u+d\in [-n,n]^k$, we have a transition
\begin{equation} (p, w(r, u), v, S, T) \xrightarrow{a}          (q, w(r, u+d), v, S, T). \label{edges:b1:cycle}\end{equation}
When we are at the end of a cycle, we have to decide whether it should be a
precise cycle or an obligation cycle. The following transition means it should be
precise: The counter effect $u$ of the cycle is added to the counter $v$, the
stack is popped, and $u$ is added to the set of repeatable effects $T$. For
each state $(p, w(p, u), v, S, T)\in Q_1$ such that $v+u\in [-B,B]^k$, we have
a transition
\begin{equation} (p, w(p, u), v, S, T) \xrightarrow{\eword}     (p, w, v+u, S, T\cup\{u\}). \label{edges:b1:precise}\end{equation}
In order to designate the cycle as an obligation cycle, we have the following
transition: The stack is popped and $u$ is added to both $S$ and $T$. For each
state $(p, w(p, u), v, S, T)\in Q_1$, we include the transition
\begin{equation} (p, w(p, u), v, S, T) \xrightarrow{\eword}     (p, w, v, S\cup\{u\}, T\cup\{u\}) \label{edges:b1:obligation}\end{equation}

The initial state is $(q_0,\eword,0,\emptyset,\emptyset)$ and the final states
are all those of the form $(q,\eword,0,S,T)$ where $q$ is final in $\A$ and
$(S,T)$ is cancellable. Employing \cref{trees} and \cref{pottier}, one can now
show that $\langof{\A}\subseteq \langof{\B_1}\subseteq\Dclosure{\langof{\A}}$.

\subparagraph{State space reduction I} We have thus shown that
$\Dclosure{\langof{\B_1}}=\Dclosure{\langof{\A}}$.  However, $\B_1$ has a
doubly exponential number of states. Therefore, we now reduce the number of
states in two steps. First, instead of remembering the set $S$ of obligation
effects, we only maintain a linearly independent set of vectors generating the
same vector space. For a set $R\subseteq \Q^k$, let $\spanof{R}$ denote the
$\Q$-vector space generated by $R$. Moreover, $\LinInd{R}$ denotes the set of
linearly independent subsets of $R$. Our new automaton $\B_2$ has states
\[ Q_2 = Q \times \Gamma^{\le n}\times [-B,B]^k\times \LinInd{[-n,n]^k}\times \Powerset{[-n,n]^k} \]
and a state in $\B_2$ is final if it is final in $\B_1$.  $\B_2$ has the same
transitions as $\B_1$, except that aside from those of type
\labelcref{edges:b1:obligation}, it has
\begin{equation} (p, w(p, u), v, S, T) \xrightarrow{\eword}     (p, w, v, S', T\cup\{u\}) \label{edges:b2:obligation}\end{equation}
for each linearly independent subset $S'\subseteq S\cup \{u\}$ such that
$\spanof{S'}=\spanof{S\cup\{u\}}$.  Of course, such an $S'$ exists for any $S$
and $u$. This means, by induction on the length, for any walk of $\B_1$ from
$(p, w, v, S, T)$ to $(q,w',v',S',T')$, we can find a walk with the same input
in $\B_2$ from $(p, w, v, S, T)$ to $(q,w',v',S'',T')$ with $S''\subseteq S'$
and $\spanof{S''}=\spanof{S'}$. Since $(S',T')$ is cancellable and $S'\subseteq
T'$, the pair $(S'',T')$ is cancellable as well. This means, our walk in $\B_2$
is accepting and hence $\langof{\B_1}\subseteq\langof{\B_2}$. It remains to verify
that $\langof{\B_2}\subseteq\langof{\B_1}$.

Observe that for any walk arriving in $(q,w,v,S,T)$ in $\B_2$, there is a
corresponding walk in $\B_1$ arriving in $(q,w,v,S',T)$ for some $S'\supseteq
S$ with $\spanof{S'}=\spanof{S}$. The next lemma tells us that if $(q,w,v,S,T)$
is a final state in $\B_2$, then $(q,w,v,S',T)$ is final in $\B_1$. This
implies that $\langof{\B_2}\subseteq\langof{\B_1}$ and hence
$\langof{\B_2}=\langof{\B_1}$.

\begin{lem}\label{span}
Let $T\subseteq \Z^k$ and $S_1\subseteq S_2\subseteq \Z^k$ such that $\spanof{S_1}=\spanof{S_2}$.
If $(S_1,T)$ is cancellable, then so is $(S_2,T)$.
\end{lem}

\subparagraph{State space reduction II} We apply a similar transformation to
the last factor of the state space. In $\B_3$, we have the state space
\[ Q_3 = Q \times \Gamma^{\le n}\times [-B,B]^k\times \LinInd{[-n,n]^k}\times \LinInd{[-n,n]^k}. \]
and a state is final in $\B_3$ if and only if it is final in $\B_2$. Analogous
to $\B_2$, we change the transitions so that instead of adding $u\in[-n,n]^k$
to $T$, we store an arbitrary $T'\in\LinInd{T\cup \{u\}}$. 

This time, it is clear that $\langof{\B_3}\subseteq\langof{\B_2}$: For every
walk in $\B_3$ arriving at $(q,w,v,S,T)$, there is a corresponding walk in
$\B_2$ arriving at $(q,w,v,S,T')$ such that $T\subseteq T'$.  Clearly, if
$(S,T)$ is cancellable, then $(S,T')$ must be cancellable as well. The
following \lcnamecref{caratheodory} implies $\langof{\B_2}\subseteq\langof{\B_3}$:
It says that for each walk in $\B_2$ arriving at $(q,w,v,S,T)$, there is 
a corresponding walk in $\B_3$ arriving at $(q,w,v,S,T')$ for some
linearly independent $T'\subseteq T$ such that $(S,T')$ is cancellable
and hence $(q,w,v,S,T')$ is final.
\begin{lem}\label{caratheodory}
Let $S,T\subseteq \Z^k$ such that $(S,T)$ is cancellable.  Then there is a
linearly independent subset $T'\subseteq T$ such that $(S,T')$ is cancellable.
\end{lem}

We have thus shown that $\Dclosure{\langof{\B_3}}=\Dclosure{\langof{\A}}$. An
estimation of the size of $Q_3$ now completes the proof of \cref{dc:blind}. 
We apply \cref{dc:blind} to derive an
algorithm for $\IncProb{\Ideal}{\fRBC}$. 
\begin{cor}\label{idealblind}
The problem $\IncProb{\Ideal}{\fRBC}$ is in $\NP$.
\end{cor}
Since \cref{dc:blind} provides an exponential bound on
$\IdealUpper{\langof{\A}}$, we can use an ideal witness $w=\Word{Y_0}^m x_1
\Word{Y_1}^m \cdots x_\ell \Word{Y_\ell}^m$ (\cref{idealwitness}) for which we have to check
membership in $\langof{\A}$. Since $\ell$ is polynomial and $m$ exponential, we
can compute a compressed representation of $w$ in form of a \emph{straight-line
program}, a context-free grammar that generates one word~\cite{Lohrey2012}.  It
follows easily from work of \citewithauthor{HagueLin2011} that membership of
such compressed words in languages of blind (or reversal-bounded) counter
automata is decidable in $\NP$.

\subparagraph{Fixed number of counters} Unfortunately, the size bound for the
NFAs provided by \cref{dc:blind} has the number of states in the exponent,
meaning that if we fix the number $k$ of counters, we still have an
exponential bound. In fact, we leave open whether one can construct
polynomial-sized NFAs for fixed $k$.  However, in many cases it suffices to have
a polynomial bound on the length of ideals.

\begin{thm}\label{idealboundrbc}
If $\A$ is an $n$-state blind $k$-counter automaton, then
$\IdealUpper{\langof{\A}}\le (5n)^{7(k+1)^2}$.
\end{thm}
Recall that an upper bound on $\IdealUpper{L}$ is essentially a pumping lemma
(see \cref{witnesses}). Here, the idea is to take a walk of $\A$ and delete cycles
until the remaining walk $u$ is at most $n$ steps. For the deleted cycles, we
take an insertion tree of height at most $n$ (\cref{trees}). Then, using
\cref{pottier}, we pick a subset $F$ (whose size is polynomial when fixing $k$)
of cycles that can balance out the effect of $u$. We then employ
\cref{pumpideal} to the insertion trees to construct an ideal whose length
is polynomial in $|F|$.

\section{Context-Free Grammars}\label{cfg}
We turn to context-free grammars. First,
we mention that given a context-free grammar $\G$, one can construct in
exponential time an (exponential-sized) NFA accepting
$\Dclosure{\langof{\A}}$~\cite{vanLeeuwen1978,Courcelle1991,Okhotin2010,GruberHolzerKutrib2009,BachmeierLuttenbergerSchlund2015}.
Second, we provide an algorithm for the problem $\IncProb{\Ideal}{\fCFG}$.
\begin{thm}\label{idealcfg}
The problem $\IncProb{\Ideal}{\fCFG}$ is in $\Poly$.
\end{thm}
\newcommand{\infsymb}[1]{a_{#1}^\omega}
In \cite{Zetzsche2015b}, this problem has been reduced to the
\emph{simultaneous unboundedness problem (SUP)} for context-free languages. The latter asks, given a
language $L\subseteq a_1^*\cdots a_n^*$, whether we have $\Dclosure{L}=a_1^*\cdots
a_n^*$. Moreover, this reduction is clearly polynomial.  Hence, we assume that
$\langof{\G}\subseteq a_1^*\cdots a_n^*$ and that the grammar $\G=(N,T,P,S)$ is productive and in Chomsky
normal form, meaning that productions are of the form $A\to BC$,
$A\to a_i$, or $A\to\eword$ for $A,B,C\in N$. First, we add productions
$A\to\eword$ for all $A\in N$, so that the resulting grammar $\G'$ satisfies
$\langof{\G'}=\Dclosure{\langof{\G}}$.  For each $A\in N$, we can in
polynomial time construct a CFG for $\{w\in (N\cup T)^* \mid
A\grammarsteps[\G']  w\}$, so we can compute the sets $L_{i}=\{A\in N \mid
A\grammarsteps[\G'] a_iA \}$ and $R_i=\{A\in N \mid A\grammarsteps[\G'] Aa_i\}$
using membership queries. We can thus compute
the grammar $\G^\omega$, which results from $\G'$ by \begin{textenum}\item
removing all productions $A\to a_i$, \item adding $A\to \infsymb{i} A$ for each
$A\in L_i$ and \item adding $A\to A\infsymb{i}$ for each $A\in R_i$.\end{textenum}
Clearly, an occurrence of $\infsymb{i}$ certifies the ability
to generate an unbounded number of $a_i$'s.  Thus, if
$\infsymb{1}\cdots\infsymb{n}\in\langof{\G^\omega}$, then $a_1^*\cdots
a_n^*\subseteq \langof{\G'}=\Dclosure{\langof{\G}}$. It is not hard to see that
the converse is true as well.
We have thus reduced the SUP to the membership problem.

\section{Algorithms}\label{algorithms}

\newcommand{\newalg}{}
\newcommand{\algtitle}[1]{\subparagraph{Algorithms for #1.}}
\algtitle{$\IncProb{\mM}{\Ideal}$} Suppose $\mM=\Ideal$ and we want to decide
whether $I\subseteq J$ for ideals $I,J\subseteq X^*$.  In logspace, we
construct an ideal witness $w$ for $I$ and $J$ (\cref{idealwitness}) and a DFA
$\A$ for $X^*\setminus J$ (\cref{idealdfa}) and check whether $w\in
\langof{\A}$. In all other cases, to decide $\Dclosure{L}\subseteq I$, we
construct a DFA $\A$ for $X^*\setminus I$ and check whether
$\Dclosure{L}\cap\langof{\A}=\emptyset$. 
\newalg
\algtitle{$\IncProb{\mM}{\fNFA}$} Suppose $\mM=\Ideal$ and we
want to decide whether $I\subseteq \Dclosure{\langof{\A}}$ for an NFA $\A$. Since
$\IdealUpper{\langof{\A}}\le |\A|$, we can construct in logspace an ideal witness
$w$ for $I$ and $\Dclosure{\langof{\A}}$ and verify $w\in\Dclosure{\langof{\A}}$. In all
other cases, we use a short witness for $\coNP$-membership.
\newalg
\algtitle{$\IncProb{\mM}{\fOCA}$} Suppose $\mM=\Ideal$  and we
want to decide whether $I\subseteq \Dclosure{\langof{\A}}$ for an OCA $\A$.  We have a polynomial bound on
$\IdealUpper{\langof{\A}}$ (see \cref{results}). Hence, we construct in logspace an ideal
witness $w$ for $I$ and $\Dclosure{\langof{\A}}$. We can also construct in logspace an OCA
$\A'$ with $\langof{\A'}=\Dclosure{\langof{\A}}$. Membership for OCA is in $\NL=\coNL$, so we can verify $w\in
I$ and $w\notin \langof{\A'}=\Dclosure{\langof{\A}}$. In all other cases, we
convert the OCA to an NFA (see \cref{results}).
\newalg
\algtitle{$\IncProb{\mM}{\fRBC[k,r]}$}
Let $\A$ be drawn from $\fRBC[k,r]$.
First, suppose $\mM=\Ideal$  and we want to decide whether $I\subseteq
\langof{\A}$.  By \cref{idealboundrbc}, we have a polynomial bound on
$\IdealUpper{\langof{\A}}$ and can construct in logspace an ideal
witness $w$ for $I$ and $\langof{\A}$. We can also construct in logspace an RBCA
$\A'$ with
$\langof{\A'}=\Dclosure{\langof{\A}}$. Since membership for $\fRBC[k,r]$ is in
$\NL$~\cite{GurariIbarra1981}, we can check whether $w\in\langof{\A'}$.
Now let $\mM\in\{\fNFA,\fOCA,\fRBC[k,r]\}$ and we are given $L$ in
$\mM$ and an automaton $\A$ from $\fRBC[k,r]$. For $\fNFA$,
$\fOCA$, and $\fRBC[k,r]$, we have a polynomial bound on $\IdealUpper{L}$ (see
\cref{results,idealboundrbc}). Thus, we guess an ideal $I$ of
polynomial length and then verify that $I\subseteq\Dclosure{L}$ but
$I\not\subseteq\Dclosure{\langof{\A}}$. Since $\IncProb{\Ideal}{\mM}$ and
$\IncProb{\Ideal}{\fRBC[k,r]}$ are in $\NL$, the verification is done in $\NL$.
Hence, non-inclusion is in $\NP$.
For $\mM\in\{\fCFG,\fRBC\}$, we assume a fixed alphabet. 
Let $L$ be in $\mM$. 
Then \cref{alphabetwitness,idealboundrbc} provide us with a witness of polynomial
length. Since (non-)membership in $\Dclosure{L}$ and in
$\Dclosure{\langof{\A}}$ can be decided in $\NP$, non-inclusion is in $\NP$.
\newalg
\algtitle{$\IncProb{\mM}{\fCFG}$}
The case $\IncProb{\Ideal}{\fCFG}$ is shown in \cref{idealcfg}.  Suppose
$\mM\in\{\fNFA,\fOCA,\fRBC[k,r]\}$ and we are given $L$ in $\mM$ and a CFG
$\G$. We have a polynomial bound on $\IdealUpper{L}$ (see
\cref{results,idealboundrbc}), so that we can guess a polynomial-length
ideal $I$. Since $\IncProb{\Ideal}{\mM}$ is in $\NL$ in every case and
$\IncProb{\Ideal}{\fCFG}$ is in $\Poly$, we can verify in polynomial time that
$I\subseteq \Dclosure{L}$ and $I\not\subseteq\Dclosure{\langof{\G}}$. Thus, non-inclusion
is in $\NP$. In the case $\mM\in\{\fRBC,\fCFG\}$, we construct exponential-sized downward closure NFAs
and check inclusion for them (and the latter problem is in $\coNP$). This yields a $\coNEXP$ algorithm.
\newalg
\algtitle{$\IncProb{\mM}{\fRBC}$}
Let $\A$ be from $\fRBC$.  The ideal case is treated in \cref{idealblind}. When given $L$ in
$\mM\in\{\fNFA,\fOCA,\fRBC[k,r]\}$, we guess a polynomial length ideal $I$
and verify that $I\subseteq\Dclosure{L}$ in $\NL$. Since
$\IncProb{\Ideal}{\fRBC}$ is in $\NP$, we can also check in $\coNP$ that
$I\not\subseteq\Dclosure{\langof{\A}}$. Hence, non-inclusion is in $\PHSigmatwo$. 
For $\mM\in\{\fCFG,\fRBC\}$, we proceed as for $\IncProb{\mM}{\fCFG}$.

\section{Hardness}\label{hardness}

In this \lcnamecref{hardness}, we prove hardness results. Most of them are
deduced from a generic hardness theorem that, under mild assumptions,
derives hardness from the ability to generate finite sets with long words.  
We will work with bounds that exhibit the following useful property.  A
monotone function $f\colon\N\to\N$ will be called \emph{amplifying} if $f(n)\ge n$ for
$n\ge 0$ and there is a polynomial $p$ such that $f(p(n))\ge f(n)^2$ for
large enough $n\in\N$.
We say that a model \emph{has property $\LowerT{f}$} (or short: \emph{is
$\LowerT{f}$})  if for each given $n\in\N$, one can construct in polynomial time a
description of a finite language whose longest word has length $f(n)$.  For
the sake of simplicity, we will abuse notation slightly and write
$\LowerT{f(n)}$ instead of $\LowerT{f}$. For a function $t\colon\N\to\N$, we
use $\coNTIME{t}$ to denote the complements of languages accepted by
nondeterministic Turing machines that are time bounded by $O(t(n^c))$ for some
constant $c$. 

We also need two mild language theoretic properties.  A \emph{transducer} is a
finite automaton where every edge reads input and produces output.  
For a transducer $\T$ and a language $L$, the language $\T L$ consists of all
words output by the transducer while reading a word from $L$.
We call a model $\mM$ a \emph{full trio model} if given a transducer $\T$ and a
language $L$ described with $\mM$, one can compute in polynomial time a
description of $\T L$. A \emph{substitution} is a map $\sigma\colon
X\to\Powerset{Y^*}$ that replaces each letter by a language. For languages $L$,
we define $\sigma(L)$ in the obvious way. We call $\sigma$ \emph{simple} if
$X\subseteq Y$ and there is some $x\in X$ such that for all $x'\in
X\setminus\{x\}$, we have $\sigma(x')=\{x'\}$ and $x$ occurs in each word from
$L$ at most once. We say that $\mM$ has \emph{closure under simple
substitutions} if given a description of $L$ and of $\sigma(x)$ in $\mM$, we can compute in
polynomial time a description of $\sigma(L)$.

\begin{thm}\label{generichardness}
Let $t\colon \N\to\N$ be amplifying and let $\mM$ and $\mN$ be full trio models
that are $\LowerT{t}$ and have closure under simple substitutions. Then both
$\IncProb{\mM}{\mN}$ and $\EqProb{\mM}{\mN}$ are hard for $\coNTIME{t}$.
Moreover, this hardness already holds for binary alphabets.
\end{thm}

Since NFAs are $\LowerT{n}$, \cref{generichardness} yields $\coNP$-hardness for
inclusion and equivalence.  In~\cite{BachmeierLuttenbergerSchlund2015},
hardness of equivalence was shown directly.  RBCA and CFG clearly exhibit
closure under simple substitutions and can generate exponentially long words.
This yields:

\begin{cor}\label{expmodels}
For $\mM,\mN\in\{\fCFG,\fRBC\}$, $\IncProb{\mM}{\mN}$ and $\EqProb{\mM}{\mN}$ are $\coNEXP$-hard.
\end{cor}

From \cref{generichardness}, we can also deduce hardness for other models.  It was shown by
\citewithauthor{HabermehlMeyerWimmel2010} that downward closures or Petri net
languages are computable, which implies decidability of our problems. We use \cref{generichardness} to prove an
Ackermann lower bound. Let $A_n\colon\N\to\N$ be defined as $A_0(x)=x+1$,
$A_{n+1}(0)=A_n(1)$, and $A_{n+1}(x+1)=A_n(A_{n+1}(x))$.  Then, the function
$A\colon\N\to\N$ with $A(n)=A_n(n)$ is the \emph{Ackermann function}.  Of
course, for large enough $n$, we have $A_n(x)\ge x^2$. For such $n$, we have
$A(n+1)=A_{n}(A_{n+1}(n))\ge A_{n+1}(n)^2\ge A(n)^2$, so $A$ is amplifying.  A
result of \citewithauthor{MayrMeyer1981} (see also \cite{PrieseWimmel2003}) states that given $n\in\N$,
one can construct in polynomial time a Petri net that, from its initial
marking, can produce up to $A(n)$ tokens in an output place. Hence,
Petri nets are $\LowerT{A}$ and they clearly satisfy the language-theoretic conditions.
\begin{cor}
For Petri net languages, inclusion and equivalence of downward closures is Ackermann-hard.
\end{cor}
 
Building
on the sufficient condition of~\cite{Zetzsche2015b},
\citewithauthor{HagueKochemsOng2016} have shown that downward closures are
computable for higher-order pushdown automata. However, the method of
\cite{Zetzsche2015b} does not yield any information about the complexity of
this computation. For $k\in\N$, we denote by $\exp_k$ the function with
$\exp_0(n)=n$ and $\exp_{k+1}(n)=2^{\exp_k(n)}$.  It is easy to see that
order-$k$ pushdown automata are $\Delta(\exp_k)$ (for instance, one can adapt
Example~2.5 of~\cite{DammGoerdt1982}).  By $\coNEXP[k]$, we denote
the complements of languages accepted by nondeterministic Turing machines in
time $O(\exp_k(n^c))$ for some constant $c$.
\begin{cor}
For higher-order pushdown automata of order $k$, inclusion and equivalence of
downward closures is hard for $\coNEXP[k]$.
\end{cor}

Our last hardness result could also be shown using the method of
\cref{generichardness}.  However, it is simpler to reduce a variant of the
subset sum problem~\cite{BeKaLaPlRy1997}.  
\begin{prop}\label{phpitwohardness}
$\IncProb{\fNFA}{\fRBC}$ and $\EqProb{\fNFA}{\fRBC}$ are $\PHPitwo$-hard, even for binary alphabets.
\end{prop}

We have thus shown hardness for all inclusion problems that do not involve
ideals. The remaining cases inherit hardness from the emptiness problem (for
$\IncProb{\mM}{\Ideal}$) or the non-emptiness problem
($\IncProb{\Ideal}{\mM}$).

\printbibliography

\appendix

\section{Ideals and Witnesses}

\begin{proof}[Proof of \cref{shortwitness}]
Let $\A=(Q,X,\Delta,q_0,Q_f)$. Consider the DFA
$\B=(\Powerset{Q},X,\Delta',Q,Q'_f)$, where from a state $P\subseteq Q$ on input
$x\in X$, we enter the state $P'$, consisting of all $q'\in Q$ that are
reachable from a state in $P$ via a path on which $x$ occurs. Moreover, $Q'_f$
is the set of all $P\subseteq Q$ with $P\cap Q_f=\emptyset$. Then clearly $\B$
accepts $X^*\setminus \Dclosure{\langof{\A}}$.

Choose $w\in\Dclosure{K}\setminus\Dclosure{\langof{\A}}$ of minimal length and
write $w=w_1\cdots w_m$ for letters $w_1,\ldots,w_m$. Suppose $m>|\A|+1$ and
consider the run of $w$ in $\B$. For each $i\in[0,m]$, let $P_i\subseteq Q$ be
the state entered after reading $w_1\cdots w_i$.  Then we have $P_0\supseteq
P_1\supseteq\cdots$ and since $m>|Q|+1$, there are $i<j$ with
$P_i=P_j$. Yet this means that also $w'=w_1\cdots w_iw_{j+1}\cdots w_m$ is a
member of $\langof{\B}=X^*\setminus\Dclosure{\langof{\A}}$. Moreover, we have $w'\subword
w$ and thus $w'\in\Dclosure{K}$. This contradicts our choice of $w$.
\end{proof}

\begin{proof}[Proof of \cref{idealwitness}]
The implications ``\labelcref{iwit:ideal}$\Rightarrow$\labelcref{iwit:every}''
and ``\labelcref{iwit:every}$\Rightarrow$\labelcref{iwit:some}'' are trivial,
so assume \labelcref{iwit:some}. Write $\Dclosure{L}=\bigcup_{i=1}^\ell I_i$ as a
union of ideals of length $\le\IdealUpper{L}$.  Then we have $\Word{Y_0}^m x_1
\Word{Y_1}^m \cdots x_n \Word{Y_n}^m\in I_i$ for some $i$. Since $I_i$ has
length at most $\IdealUpper{L}$, there is an NFA $\A$ with at most
$\IdealUpper{L}+1$ states for $I_i$. However, we have $m\ge\IdealUpper{L}+1$, so in
the computation of the NFA for $\Word{Y_0}^m x_1 \Word{Y_1}^m \cdots x_n
\Word{Y_n}^m$, for each $i\in[0,n]$, some power $\Word{Y_i}^{k_i}$, $k_i>0$,
has to lie on a cycle of $\A$. We can therefore pump each of these cycles,
which implies $I\subseteq I_i\subseteq \Dclosure{L}$.
\end{proof}

\begin{proof}[Proof of \cref{idealdfa}]
Let $I=Y_0^*\{x_1,\varepsilon\}Y_1^*\cdots\{x_n,\varepsilon\}Y_n^*$ with.   For $i\in [0,n]$ and $a\in X$, let 
\[ J_{i,a}=\{j\in [i,n] \mid a\in Y_i^*\{x_{i+1},\varepsilon\}Y_{i+1}^*\cdots \{x_{j},\varepsilon\}Y_{j}^* \} \]
Our DFA has states $Q=\{0,\ldots,n+1\}$ and for $i\in Q$, we have $i\xrightarrow{a} j$ if and only if
\[ j=\begin{cases} \min J_{i,a} & \text{if $J_{i,a}\ne\emptyset$} \\ n+1 & \text{if  $J_{i,a}=\emptyset$}\end{cases} \]
Moreover, $0$ is the initial state and the states $0,\ldots,n$ are final.
Clearly, the automaton is ordered, has $n+2$ states, and can be constructed in
logarithmic space.  In order to show the correctness, we define the ideal
$I_k=Y_0^*\{x_1,\varepsilon\}Y_1^*\cdots\{x_k,\varepsilon\}Y_k^*$ for each
$k\in[0,n]$. Observe that $I_0\subseteq I_1\subseteq \cdots \subseteq I_n=I$.
By induction on the length of $w$, it is easy to see that if
$0\xrightarrow{w}j$, then
\begin{itemize}
\item if $j\in[0,n]$, then $j$ is the smallest number with $w\in I_j$.
\item if $j=n+1$, then $w\notin I$.
\end{itemize}
In particular, the automaton accepts $I$.
\end{proof}

\begin{proof}[Proof of \cref{ordereddfacycles}]
To make our induction work, we define $f_n\colon \N\to\N$ by $f_n(1)=n-1$
and $f_n(k)=(f_n(k-1)+1)\cdot (n-1)$. We claim that if $w>f_n(|X|)$ for $w\in X^*$,
then $w$ has a position at which every ordered $n$-state DFA cycles.

We proceed by induction on $|X|$. If $X=\{a\}$, then it suffices to consider
$w=a^{n}$.  Consider an ordered $n$-state DFA $\A$ and let $q_0,q_1,\ldots,q_n$
be the states occupied while reading $w$. Then there are $i<j$ with $q_i=q_j$
and since $\A$ is ordered, we have $q_i=q_{i+1}$. This means, $q_i$ has an
$a$-labeled loop and therefore $q_i=q_{i+1}=q_{i+2}=\cdots =q_n$. In
particular, $\A$ cycles at the last position of $w$.

Now suppose $k=|X|>1$ and $|w|>f_n(k)=(f_n(k-1)+1)(n-1)$. For every word $v\in X^*$,
let $\alpha(v)\in X^*$ be the shortest prefix of $v$ in which every letter from
$X$ occurs.  If $v$ does not contain every letter from $X$, then we define
$\alpha(v)=v$.  We factorize $w$ as $p_1\cdots p_m$ by applying $\alpha$ to
$w$, then applying $\alpha$ to the rest of the word, and so on. Formally, we
set $r_0=w$, $p_i=\alpha(r_{i-1})$, and define $r_i$ so that $r_{i-1}=p_ir_i$.
For some smallest $m\ge 1$, we have $p_m=r_m$. Then clearly $w=p_1\cdots p_m$
and every $p_i$ is non-empty.

For each $i\in[1,m]$, let $p'_i$ be obtained from $p_i$ by removing its last
position. By the choice of $p_i$, the word $p'_i$ contains at most $k-1$
distinct letters. Hence, if $|p'_i|>f_n(k-1)$ for some $i\in[1,m]$, then $p'_i$
contains a position at which every ordered $n$-state DFA cycles. In particular,
$w$ contains such a position (because every computation on $w$ contains some
computation on $p'_i$).  Therefore, we may assume that $|p_i|=|p'_i|+1\le
f_n(k-1)+1$ for every $i\in[1,m]$.

If we had $m\le n-1$, this would imply $|w|=|p_1\cdots p_m|\le
(f_n(k-1)+1)(n-1)$, which is not the case.  Hence, we have $m\ge n$.
Now consider an ordered $n$-state DFA $\A$ with its computation
\[q_0\xrightarrow{p_1} q_1\xrightarrow{p_2} \cdots \xrightarrow{p_m} q_m. \]
Since $m+1>n$, there are $i<j$ with $q_i=q_j$ and since $\A$ is ordered,
we have $q_i=q_{i+1}=\cdots=q_j$. We distinguish two cases.
\begin{itemize}
\item If $j=m$, then our computation cycles at every position in $p_m$.
\item If $j<m$, then $p_j$ contains every letter from $X$ at least once. This
means $q_i$ has an $a$-loop for every $a\in X$. Therefore,
$q_i=q_{i+1}=\cdots=q_m$. In particular, our computation cycles on every
position in $p_m$.
\end{itemize}
Thus, we have shown that any ordered $n$-state DFA cycles on every position in $p_m$,
which proves our claim.

From the definition of $f_n$, it follows easily by induction that $f_n(k)=\sum_{i=1}^k (n-1)^i$
and hence $f_n(k)\le k\cdot (n-1)^k$.
\end{proof}

\section{Insertion trees}

\begin{proof}[Proof of \cref{trees}]
Let $w\in\Delta^*$ be a prime $q$-cycle and let $P_w\subseteq Q$ be the set of
states occurring properly in $w$. We show by induction on $|P_w|$ that every
prime cycle $w$ admits an insertion tree of height at most $|P_w|$.

If no state from $P_w$ repeats in $w$, then $w$ is simple and the statement is
trivial. For each $p\in P_w$ that does repeat in $w$, let $\lambda(p)$ be the
length of the longest $p$-cycle that is a factor of $w$. Among all states from
$P_w$ that repeat in $w$, we choose $p$ such that $\lambda(p)$ is maximal. Then
$w=xyz$ where $y$ is a $p$-cycle of length $\lambda(p)$. Observe that by the
maximality of $p$, there is no state that occurs properly both in $x$ and in
$z$.

We write $y=y_1\cdots y_r$ such that each $y_i$ is a prime $p$-cycle. Then
since $p$ does not occur properly in $y_i$, each $y_i$ admits an insertion tree
$t_i$ of height $|P_w|-1$.

Consider any insertion tree $t$ of $xz$. Observe that since there is no
state that occurs properly both in $x$ and in $z$, the only cycle in $t$ where
$p$ can occur is $t$'s root. Therefore, if $s_1,\ldots,s_k$ are the subtrees of
$t$ immediately below the root, then no $\alpha(s_i)$ contains $p$. We can
therefore factorize each $\alpha(s_i)$ into prime cycles that each have an insertion tree 
of height at most $|P_w|-1$.  Thus, by replacing in $t$ each
$s_i$ by this sequence of trees, we obtain an insertion tree $t'$ of $xz$ of
height at most $|P_w|$.

Since $p$ occurs in the root of $t'$ and this is the only occurrence of $p$ in
$t'$, we can attach the trees $t_i$ directly below the root of $t'$ to obtain a
insertion tree $t''$ of $w$. Moreover, since each $t_i$ has height at
most $|P_w|-1$, $t''$ has height at most $|P_w|$.
\end{proof}

\begin{proof}[Proof of \cref{pumpideal}]
If $F=\emptyset$, then we can duplicate every tree in the sequence, leading to
$\Dclosure{\PumpLang{s}{F}}=Y^*$, where $Y$ is the set of letters occurring
anywhere on a tree in $s$. Hence, $\Dclosure{\PumpLang{s}{F}}$ is an ideal of
expression length one.  Thus, we assume $F\ne\emptyset$.

As the first step, we consider the case where $s$ consists of one tree $t$.
Let $A$ be the set of vertices in $t$ that are ancestors of vertices in $F$. We
show by induction on $h$ that $\Dclosure{\PumpLang{t}{F}}$ is an ideal
and $\Elen{\PumpLang{t}{F}}\le |A|\cdot (2|Q|+|F|)$.

Let $r$ be the root of $t$ and $\gamma(r)=e_1\cdots e_\ell$, where
$e_1,\ldots,e_\ell\in\Delta$.  Let $C$ be the set of children of $r$ that are
in $A$. Moreover, let $e_i=(q_{i-1}, a_i, q_i)$ for $i\in[1,\ell-1]$. Recall
that every child of $r$ is assigned a $q_i$-cycle for some $i\in[1,\ell-1]$.
For each $i\in [1,\ell-1]$, consider the subtrees `inserted after $e_i$': In
other words, those subtrees directly below $r$ whose root node is assigned a
$q_i$-cycle by $\gamma$. Some of them contain a fixed vertex; let
$s_{i,1},\ldots,s_{i,k_i}$ be those subtrees. The other subtrees inserted after
$e_i$ are pumpable; let $Y_i$ be the set of input letters occurring in them.
Let $F_{i,j}\subseteq F$ be the set of fixed nodes in $s_{i,j}$. Moreover, let
$A_{i,j}$ be the set of vertices in $s_{i,j}$ that are ancestors of fixed
vertices (in $s_{i,j}$). Note that since $F\ne\emptyset$, we have $r\in A$ and thus
\[ |A|=1+\sum_{i=1}^{\ell-1} \sum_{j=1}^{k_i} |A_{i,j}|. \]
By induction, $\Dclosure{\PumpLang{s_{i,j}}{F_{i,j}}}$ is an ideal and we have
\begin{equation} \Elen{\Dclosure{\PumpLang{s_{i,j}}{F_{i,j}}}} \le |A_{i,j}|\cdot (2|Q|+|F_{i,j}|)\le |A_{i,j}|\cdot (2|Q|+|F|).\label{ideallength:induction}\end{equation}
\begin{itemize}
\item Suppose $r\in F$. Then we have
$\Dclosure{\PumpLang{t}{F}}=\{a_1,\eword\}I_1\{a_2,\eword\}\cdots I_{\ell-1}\{a_\ell,\eword\}$,
where
\[ I_i = Y_i^*(\Dclosure{\PumpLang{s_{i,1}}{F_{i,1}}})Y_i^*\cdots (\Dclosure{\PumpLang{s_{i,k_i}}{F_{i,k_i}}})Y_i^* \]
for $i\in[1,\ell-1]$. Hence, $\Dclosure{\PumpLang{t}{F}}$ is an ideal. Let us estimate the expression length.
Note that \eqref{ideallength:induction} yields
\[ \Elen{I_i} \le  k_i+1+\sum_{j=1}^{k_i} \Elen{\Dclosure{\PumpLang{s_{i,j}}{F_{i,j}}}}
\le k_i + 1 + (2|Q|+|F|)\cdot \sum_{j=1}^{k_i} |A_{i,j}| \]
and therefore
\begin{align*}
\Elen{\Dclosure{\PumpLang{t}{F}}} & \le \ell+\sum_{i=1}^{\ell-1} |I_i| \le \ell + \sum_{i=1}^{\ell-1} (k_i+1) + (2|Q|+|F|)\sum_{i=1}^{\ell-1}\sum_{j=1}^{k_i}|A_{i,j}| \\
&\le 2\ell + \underbrace{\sum_{i=1}^{\ell-1} k_i}_{\le |F|} ~~+~~ (2|Q|+|F|)\underbrace{\sum_{i=1}^{\ell-1}\sum_{j=1}^{k_i}|A_{i,j}|}_{=|A|-1} \\
&\le 2|Q|+|F| ~~+~~(2|Q|+|F|)\cdot (|A|-1) \\
&\le |A|\cdot (2|Q|+|F|).
\end{align*}
\item Suppose $r\notin F$. Then we have $\Dclosure{\PumpLang{t}{F}}=I_1\cdots I_{\ell-1}$, where
\[ I_i=Z_i^*(\Dclosure{\PumpLang{s_{i,1}}{F_{i,1}}})Z_i^*\cdots (\Dclosure{\PumpLang{s_{i,k_i}}{F_{i,k_i}}})Z_i^*, \]
for $i\in[1,\ell-1]$ with $Z_i=Y_i\cup \{a_1,\ldots,a_\ell\}$. Hence, $\Dclosure{\PumpLang{t}{F}}$ is an ideal. Let us estimate the expression length.
As before, \eqref{ideallength:induction} yields
\[ \Elen{I_i}\le k_i + 1 + \sum_{i=1}^{k_i} \Elen{\Dclosure{\PumpLang{s_{i,j}}{F_{i,j}}}} \le k_i + 1 + (2|Q|+|F|)\cdot \sum_{j=1}^{k_i} |A_{i,j}| \]
and therefore
\begin{align*}
\Elen{\Dclosure{\PumpLang{t}{F}}} &\le \sum_{i=1}^{\ell-1} (k_i+1) ~~+~~ (2|Q|+|F|)\sum_{i=1}^{\ell-1}\sum_{j=1}^{k_i}|A_{i,j}| \\
&\le \ell+\sum_{i=1}^{\ell-1} k_i ~~+~~ (2|Q|+|F|)\sum_{i=1}^{\ell-1}\sum_{j=1}^{k_i}|A_{i,j}| \\
&\le 2|Q|+|F|~~+~~(2|Q|+|F|)\cdot (|A|-1) \\
&\le |A|\cdot (2|Q|+|F|).
\end{align*}
\end{itemize}

This concludes our first step. Note that since $t$ has height $\le h$, every
vertex in $F$ has at most $h$ ancestors, so that $|A|\le h\cdot |F|$.  This
means, our first step implies that in the case of a single tree $t$, we have
$\Elen{\Dclosure{\PumpLang{t}{F}}}\le h|F|\cdot(2|Q|+|F|)$.

Let us now consider $\Dclosure{\PumpLang{s}{F}}$ where $s=t_1\cdots t_m$ is a
compatible sequence. Of the trees $t_1,\ldots,t_m$, let $t'_1,\ldots,t'_\ell$
be those which contain a fixed vertex. The other trees in the sequence
$t_1,\ldots,t_m$ are pumpable and we define $Y$ to be the set of letters
occurring in those pumpable trees. Note that $\ell\le |F|$.

According to our first step, we have
$\Elen{\Dclosure{\PumpLang{t'_i}{F}}}\le h|F|\cdot(2|Q|+|F|)$
for each $i\in[1,\ell]$. Moreover, we have
\[ \Dclosure{\PumpLang{s}{F}}=Y^*\Dclosure{\PumpLang{t'_1}{F}}Y^*\cdots \Dclosure{\PumpLang{t'_\ell}{F}}Y^*, \]
which means $\Dclosure{\PumpLang{s}{F}}$ is an ideal and we may estimate
\begin{align*}
\Elen{\Dclosure{\PumpLang{s}{F}}} &\le (\ell+1)+\sum_{i=1}^\ell \Elen{\Dclosure{\PumpLang{t'_i}{F}}} \le \ell+1+\ell\cdot h|F|\cdot (2|Q|+|F|) \\
&\le h\cdot |F|\cdot (2|Q|+|F|)^2,
\end{align*}
which proves the \lcnamecref{pumpideal}.
\end{proof}

\section{Counter automata}

We prove the statements of \cref{counterautomata} in the order they are made.
We begin with \cref{dc:blind}.
\begin{proof}[Proof of \cref{dc:blind}]
We have seen that $\Dclosure{\langof{\B_3}}=\Dclosure{\langof{\A}}$. To
estimate the size of $\B_3$, notice that 
\begin{align*}
|[-B,B]^k|&=(2B+1)^k\le (2(n+n\cdot(3n)^{(k+1)^2})+1)^k \\
&=(2n\cdot (3n)^{(k+1)^2}+2n+1)^k \le (3n)^{((k+1)^2+1)k} \le (3n)^{5k^3}.
\end{align*}
Furthermore, our stack alphabet satisfies
$|\Gamma|=n\cdot (2n+1)^k$, so that 
\[ |\Gamma^{\le n}|\le |\Gamma|^{n+1}=n^{n+1}\cdot (2n+1)^{(n+1)k}\le (3n)^{4nk}. \]
Finally, we can estimate $|\LinInd{[-n,n]^k}|\le \binom{(2n+1)^k}{k}\le
(2n+1)^{k^2}\le (3n)^{k^2}$.  This means in total $|Q_3|\le n\cdot
(3n)^{4nk+5k^3+2k^2}\le (3n)^{5nk+7k^3}$. This completes the proof of
\cref{dc:blind}.
\end{proof}

Next, we show that $\langof{\B_1}$ and $\langof{\A}$ have the same downward closure.
\begin{prop}\label{counter:b1}
$\langof{\A}\subseteq \langof{\B_1}\subseteq\Dclosure{\langof{\A}}$.
\end{prop}

We prove \cref{counter:b1} in the following two lemmas.
\begin{lem}
$\langof{\A}\subseteq\langof{\B_1}$.
\end{lem}
\begin{proof}
Let $w\in \Delta^*$ be an accepting walk of $\A$. We can write $w=u_0 v_1 u_1 \cdots
v_\ell u_\ell$ such that $|u_0\cdots u_\ell|\le n$ and every $v_i$ is a prime
cycle. For each $v_i$, \cref{trees} yields an insertion tree $t_i$ of height
at most $n$.  In $\B_1$, we simulate $u_0\cdots u_\ell$ by transitions of type
\labelcref{edges:b1:ord}.  When we arrive at a prime cycle $v_i$, we traverse
the tree $t_i$: When at the current state a subtree is attached in $t_i$, we
use a transition of type \labelcref{edges:b1:push}. When we arrive at the state
where our current cycle has started, we use either \labelcref{edges:b1:precise}
or \labelcref{edges:b1:obligation} to use the cycle as a precise cycle or as an
obligation cycle, respectively. During a cycle, we use transitions
\labelcref{edges:b1:cycle}.

It remains to be shown that there exists a choice of cycles as `precise' or
`obligation' to obtain an accepting run of $\B_1$, i.e. the capacity in the
factor $[-B,B]^k$ is not exceeded and the sets $S$ and $T$ in the factors
$\Powerset{[-n,n]^k}$ form a cancellable $(S,T)$. To this end, we apply
\cref{pottier}. Let $e_1,\ldots,e_m\in\Z^k$ be the different effects (in any
order) of the (simple) cycles in all the insertion trees $t_i$,
$i\in[1,\ell]$.  Being effects of simple cycles, they are even contained in
$[-n,n]^k$.  For each $i\in[1,m]$, let $x_i$ be the number of times $e_i$
occurs as an effect of a cycle.  Let $e$ be the effect of the walk $u_0\cdots
u_\ell$. Then we have $e\in[-n,n]^k$.  Since the walk $w=u_0 v_1 u_1 \cdots
v_\ell u_\ell$ is accepting in $\A$, we have $e+\sum_{i=1}^m x_ie_i=0$. 

Consider the matrix $A\in\Z^{k\times m}$ with columns $e_1,\ldots,e_m$.  Then the vector
$x=(x_1,\ldots,x_m)$ satisfies $Ax=-e$. Since the $e_i$ are pairwise distinct and
members of $[-n,n]^k$, we have $m\le (2n+1)^k$. This yields
$\|(A|e)\|_{1,\infty}\le (m+1)n$. Moreover, $A$ has rank at
most $k$.  By \cref{pottier}, there exists a $y\in\N^m$ with $Ay=-e$, $y\le x$,
and 
\[
\|y\|_1\le (1+(m+1)n)^{k+1}=(mn+n+1)^{k+1}\le ((2n+1)^{k+1})^{k+1}\le (3n)^{(k+1)^2}.
\]
We can therefore choose for each $i\in[1,m]$, $y_i$ of the $x_i$ cycles with
effect $e_i$ and use them as precise cycles. Then, in the end, we arrive at a
state $(q,\eword,v,S,T)$ with $v=0$. Since we used at most $\|y\|_1$ precise
cycles and at most $n$ transitions in the walk $u_1\cdots u_\ell$, the counter
values encountered during the computation are bounded in absolute value by
$n+n\cdot (3n)^{(k+1)^2}=B$.

Observe that we have $T=\{e_1,\ldots,e_m\}$. Consider $z\in\N^m$ with $z=x-y$.
By our choice of precise cycles, $S=\{e_i \mid z_i>0\}$. Therefore, since
$Az=0$, the pair $(S,T)$ is cancellable. Hence, we have reached a final state of
$\B_1$ and read the same word as $w$. 
\end{proof}

\begin{lem}
$\langof{\B_1}\subseteq \Dclosure{\langof{\A}}$.
\end{lem}
\begin{proof}
Consider a walk $w$ in $\B_1$ from $(p,\eword,0,S,T)$ to $(q,\eword,v,S',T')$.
Let $S'\setminus S=\{e_1,\ldots,e_m\}$ and $(S'\setminus S)\cup (T'\setminus
T)=\{e_1,\ldots,e_{m+\ell}\}$.  Let $A\in \Z^{(m+\ell)\times k}$ be the matrix
with columns $e_1,\ldots,e_{m+\ell}$.  Moreover, for each $i\in[1,m]$, let
$x_i\in\N\setminus\{0\}$ be the number of times a cycle with effect $e_i$ was
used as an obligation cycle. Let $x\in\N^{m+\ell}$ be the vector
$x=(x_1,\ldots,x_{m+\ell})$ where $x_{m+i}=0$ for $i\in[1,\ell]$.

It is easy to show by induction on the maximal stack height in
$w$ that for every $y\in\N^{m+\ell}$ with $y\ge x$, there exists a walk $w'$ in
$\A$ from $(p,0)$ to $(q,v+Ay)$ such that $w'$ reads a superword of the input
of $w$: We execute all the obligation cycles as normal cycles in $\A$, which
means adding the effect $Ax$. Then, for each effect $e_i$, we execute some
cycle with effect $e_i$ an additional $y_i-x_i$ times. In total, we add $v+Ay$
to the counter in $\A$. 

Now suppose $w$ is an accepting walk. Then $S=T=\emptyset$, the pair $(S',T')$
is cancellable, and $v=0$. Since $(S',T')$ is cancellable, there is a
$z\in\N^{m+\ell}$ with $z_i\ge 1$ for $i\in[1,m]$ such that $Az=0$. Since
$x_{m+i}=0$ for $i\in[1,\ell]$, we can find a number $M\in\N\setminus\{0\}$
such that $Mz\ge x$. We set $y=Mz$ and since then $y\ge x$, we may apply our
observation above to this $y$.  This yields a walk $w'$ in $\A$ from $(p,0)$ to
$(q,v+Ay)=(q,0+MAz)=(q,0)$ such that $w'$ reads a superword of the word read by
$w$. This means, $w'$ is accepting, so that the word read by $w$ is contained
in $\Dclosure{\langof{\A}}$.
\end{proof}

\begin{proof}[Proof of \cref{span}]
We may clearly assume that  $|S_2\setminus S_1|=1$.
Hence, let $(S_1,T)$ be cancellable, $S_1=\{u_1,\ldots,u_s\}$, and $S_2=\{u_1,\ldots,u_{s+1}\}$.
Since $\spanof{S_1}=\spanof{S_2}$, there are $z_1,\ldots,z_{s+1}\in\Q$ with $z_{s+1}\ne 0$ and
$\sum_{i=1}^{s+1} z_iu_i=0$. By multiplying with a common denominator and,
if necessary, switching the sign of the $z_1,\ldots,z_{s+1}$, we may assume that
$z_1,\ldots,z_s\in\Z$ and $z_{s+1}\in\N\setminus\{0\}$.

Let $T=\{v_1,\ldots,v_t\}$. Since $(S_1,T)$ is cancellable, there are are
$x_1,\ldots,x_s\in\N\setminus\{0\}$ and $y_1,\ldots,y_t\in\N$ with
\[ x_1u_1+\cdots +x_su_s ~~+~~y_1v_1+\cdots+y_tv_t=0. \]
Since $x_i\ge 1$ for $i\in[1,s]$, we can find $M\in\N\setminus\{0\}$ with
$M\cdot x_i > -z_i$ for every $i\in[1,s]$. Then, since $\sum_{i=1}^{s+1} z_iu_i=0$, we have
\begin{align*}
0 &= M\left(\sum_{i=1}^s x_i u_i+\sum_{i=1}^t y_iv_i\right) + \sum_{i=1}^{s+1} z_i u_i = \sum_{i=1}^s (Mx_i+z_i)u_i + z_{s+1}u_{s+1} + \sum_{i=1}^t (My_i)v_i
\end{align*}
Since $M x_i+z_i\in\N\setminus\{0\}$ for $i\in[1,s]$ and
$z_{s+1}\in\N\setminus\{0\}$, this proves that $(S_2,T)$ is cancellable.
\end{proof}

\begin{proof}[Proof of \cref{caratheodory}]
Let $S=\{u_1,\ldots,u_s\}$ and choose $T'\subseteq T$ minimal with the property
that $(S,T')$ is cancellable.  Let $T'=\{v_1,\ldots,v_t\}$. Then there are
$x_1,\ldots,x_s\in\N\setminus\{0\}$ and $y_1,\ldots,y_t\in\N$ with
$\sum_{i=1}^s x_iu_i + \sum_{i=1}^t y_iv_i=0$. By minimality of $T'$, we have
$y_i>0$ for every $i\in[1,t]$.  Suppose $T'$ is linearly dependent. Then there
are $z_1,\ldots,z_t\in\Z$, not all zero, such that $\sum_{i=1}^t z_iv_i=0$. We
may assume that at least one $z_i$ is positive, because otherwise they are all
at most zero and we can negate them.

Choose $j\in[1,t]$ such that $z_j/y_j$ is maximal, meaning $z_j/y_j\ge z_i/y_i$
for every $i\in[1,t]$. Note that then $z_j>0$ because otherwise, $z_i\le 0$ for
every $i\in[1,t]$. Then we have $z_jy_i\ge y_jz_i$ for every $i\in[1,t]$ and
hence
\begin{align*}
0 = z_j\left(\sum_{i=1}^s x_iu_i+\sum_{i=1}^t y_iv_i\right)-y_j\left(\sum_{i=1}^t z_iv_i\right)= \sum_{i=1}^s (z_jx_i)u_i + \sum_{i=1}^t \underbrace{(z_jy_i-y_jz_i)}_{\ge 0}v_i.
\end{align*}
Since $z_j x_i>0$ for $i\in[1,s]$ and we have the coefficient $z_jy_j-y_jz_j=0$
in front of $v_j$, the last equation tells us that $(S,T'\setminus\{v_j\})$ is
cancellable. This contradicts the choice of $T'$. Therefore $T'$ is linearly
independent.
\end{proof}

\begin{proof}[Proof of \cref{idealblind}]
Suppose we are given an ideal $I=Y_0^* \{x_1,\eword\} Y_1^* \cdots
\{x_\ell,\eword\} Y_\ell^*$ and a blind $k$-counter automaton $\A$ with $n$
states. By \cref{dc:blind}, we have an exponential bound upper bound $m$ on
$\IdealUpper{\langof{\A}}$.  According to \cref{idealwitness}, we have
$I\subseteq\Dclosure{\langof{\A}}$ if and only if $w:=\Word{Y_0}^m x_1
\Word{Y_1}^m \cdots x_\ell \Word{Y_\ell}^m\in \Dclosure{L}$. Now the word $w$
may be exponentially long, but since $m$ is at most exponential, we can compute
$m$ in binary representation. 

Given the polynomial-sized ideal and the binary representation of $m$, we can
construct a polynomial-sized straight-line program $\G$ for $w$: A
\emph{straight-line program} (SLP) is a context-free grammar that generates
exactly one word (see \cite{Lohrey2012} for details and a survey). $\G$ is
obtained from an SLP for the polynomial-length word $z_0 x_1 z_1 \cdots x_\ell
z_\ell$ and SLPs for the words $w_{Y_i}^m$, which in turn result from an SLP
for $\{a^m\}$. The latter is easily constructed from the binary representation
of $m$.

Therefore, it remains to be shown that the  \emph{compressed membership
problem} for blind counter automata is decidable in $\NP$. The latter asks,
given an SLP $\G$ and a blind counter automaton $\A$, whether the word
generated by $\G$ is accepted by $\A$. This can be decided by constructing
an automaton that has access to a pushdown and blind (or reversal-bounded)
counters that accepts $\langof{\G}\cap\langof{\A}$. For such automata,
the emptiness problem is in $\NP$, as shown by \citewithauthor{HagueLin2011}.
\end{proof}

\begin{proof}[Proof of \cref{idealboundrbc}]
We show that for every $u\in \langof{\A}$, there exists an ideal $I$ of length
at most $(5n)^{7(k+1)^2}$ with $u\in I\subseteq \Dclosure{L}$.

So let $w\in \Delta^*$ be an accepting walk of $\A$. We can write $w=u_0 v_1
u_1 \cdots v_\ell u_\ell$ such that $u_0,\ldots,u_\ell\in\Delta$, $|u_0\cdots u_\ell|\le n$,
and every $v_i$ is a cycle. We factorize each $v_i=v_{i,1}\cdots v_{i,k_i}$
into prime cycles $v_{i,1},\ldots,v_{i,k_i}$ and let \cref{trees} provide an
insertion tree $t_{i,j}$ of $v_{i,j}$ of height at most $n$.

Let $e_1,\ldots,e_m\in\Z^k$ be the effects of cycles occurring in any of these
trees. Note that $\|e_i\|_\infty\le n$ for $i\in[1,m]$, so that $m\le (3n)^k$.
For $i\in[1,m]$, let $x=(x_1,\ldots,x_m)\in \N^m$ be the vector such
that $x_i$ is the number of times a cycle with effect $e_i$ occurs. Moreover,
let $e\in\Z^k$ be the effect of $u_0\cdots u_\ell$. Then we have $\|e\|_\infty\le n$.
Let $A\in\Z^{k\times m}$ be
the matrix with columns $e_1,\ldots,e_m$.  Since $w$ is accepting, we have
$Ax=-e$. Note that
\[ \| (A|e) \|_{1,\infty} \le \|e\|_\infty + \sum_{i=1}^m \|e_i\|_{\infty}\le (m+1)n\le ((3n)^k + 1)n\le (4n)^{k+1}\]
and that the rank  of $(A|e)$ is at most $k$. According to \cref{pottier},
there is a $y\in\N^m$, $y\le x$, such that $Ay=-e$ and $\|y\|_1\le (1+(4n)^{k+1})^{k+1}\le (5n)^{(k+1)^2}$.

From our insertion trees, we now select for each $i\in[1,m]$, $y_i$-many
vertices whose cycles have effect $e_i$. This is possible since $y\le x$.  Let
$F$ be the set of these vertices.  Then we have
$|F|\le \|y\|_1\le (5n)^{(k+1)^2}$. For each $i\in[0,\ell]$, let $a_i\in X\cup \{\eword\}$ be the input read by $u_i$.
We claim that the language
\[ K=a_0 \PumpLang{t_{1,1}\cdots t_{1,k_i}}{F} a_1 \cdots \PumpLang{t_{\ell,1}\cdots t_{\ell,k_\ell}}{F} a_\ell \]
is contained in $\Dclosure{\langof{\A}}$. Let $z=(z_1,\ldots,z_m)\in\N^m$ be
the vector with $z=x-y$. Then $Az=0$ and every pumpable vertex has an effect
$e_i$ where $z_i\ge 1$.

Now suppose we obtain a walk $w'$ of $\A$ by performing some pumping to obtain a
word $u'\in K$, either by duplicating a single vertex or by duplicating a whole
pumpable subtree.  Note that it might happen that $w'$ does not leave the
counters at zero in the end. But we will show that we can pump even more to get
such a walk.  For each $i\in[1,m]$, let $z'_i\in\N$ be the number of times we
add an occurrence of a cycle with effect $e_i$.  Let $z'=(z'_1,\ldots,z'_m)$.
Since $z'_i\ge 1$ implies $z_i\ge 1$, we can find an $N\in\N$ with $N\cdot z\ge
z'$. Now for every $i\in[1,m]$ with $z'_i\ge 1$, we can find a pumpable vertex $v_i$ whose cycle
has effect $e_i$. We can pump $v_i$ an additional $Nz_i-z'_i$ times. This
results in a walk $w''$ of $\A$ with effect $e+Ax+ANz=e+Ax=0$, meaning that it is accepting.  Moreover, if
$u''$ is the input word read by $w''$, then we have $u'\subword
u''\in\langof{\A}$. This proves $K\subseteq\Dclosure{\langof{\A}}$, which was
our claim.

This means that the language
\[ I=\Dclosure{K}=\{ a_0, \eword\} \Dclosure{\PumpLang{t_{1,1}\cdots t_{1,k_1}}{F}} \{ a_1,\eword\} \cdots \Dclosure{\PumpLang{t_{\ell,1}\cdots t_{\ell,k_\ell}}{F}} \{ a_\ell,\eword\} \]
is contained in $\Dclosure{\langof{\A}}$. By \cref{pumpideal}, it is an ideal and satisfies
\begin{align*}
\Elen{I}&\le \ell+1+\sum_{i=1}^\ell \Elen{\Dclosure{\PumpLang{t_{i,1}\cdots t_{i,k_i}}{F}}}\le n+n^2\cdot |F|\cdot (2n+|F|)^2 \\
&\le \left(2n^2\cdot (5n)^{(k+1)^2}\right) \cdot \left(2n+(5n)^{(k+1)^2}\right)^2 \\
&\le \left((5n)^{3(k+1)^2}\right) \cdot (5n)^{4(k+1)^2}\le (5n)^{7(k+1)^2},
\end{align*}
which completes our proof.
\end{proof}

\section{Context-Free Grammars}

The following lemma remains to be shown.
\begin{lem}\label{cfg:correctness}
We have $\infsymb{1}\cdots\infsymb{n}\in \langof{\G^\omega}$ if and only if $a_1^*\cdots a_n^*\subseteq \Dclosure{\langof{\G}}$.
\end{lem}
\begin{proof}
Suppose $a_1^*\cdots a_n^*\subseteq \Dclosure{\langof{\G}}=\langof{\G'}$. Then
there are derivation trees $t_1,t_2,\ldots$ of $\G'$ with
$|\yield{t_j}|_{a_i}\ge j$ for every $i\in[1,n]$ and $j\ge 1$.

On the vertices of $t_j$, we define partial orders $\ll_i$ as follows. We have
$u\ll_i v$ if $v$ is a descendant of $u$ and
$|\yield{u}|_{a_i}>|\yield{v}|_{a_i}$. By induction on $\ell$, it is easy to
check that if all $\ll_i$-chains in $t_j$ have length $\le\ell$, then
$|\yield{t_j}|_{a_i}\le 2^\ell$. Hence, if $m>2^{|N|}$, then
$|\yield{t_m}|_{a_i}> 2^{|N|}$, so that $t_m$ must have a $\ll_i$-chain of
length $>|N|$.  On this chain, some $A_i\in N$ has to repeat, meaning $A_i\in
L_i\cup R_i$.  We can therefore expand $t_m$ by applying for each $i\in[1,n]$
the production $A_i\to \infsymb{i}A_i$ or $A_i\to A_i\infsymb{i}$. Then, we
replace every $a_i$-leaf by $\eword$. By construction, the resulting tree  $t$
is a derivation tree of $\G^\omega$ and every $\infsymb{i}$ appears exactly
once. Hence, $\yield{t}$ is a permutation of $\infsymb{1}\cdots\infsymb{n}$.
It remains to be shown that $\yield{t}=\infsymb{1}\cdots\infsymb{n}$.

Consider the morphism $\alpha\colon
\{\infsymb{1},\ldots,\infsymb{n}\}^*\to\{a_1,\ldots,a_n\}^*$ such that for
every $i\in[1,n]$, we have $\alpha(\infsymb{i})=a_i$. Recall that for every
production $A\to\infsymb{i}A$ or $A\to A\infsymb{i}$ in $\G^\omega$, we have
$A\grammarsteps[\G'] a_iA$ or $A\grammarsteps[\G'] Aa_i$, respectively.  This
tells us that $\alpha(\langof{\G^\omega})\subseteq \langof{\G'}\subseteq
a_1^*\cdots a_n^*$. Therefore, $\alpha(\yield{t})=a_1\cdots a_n$ and hence
$\yield{t}=\infsymb{1}\cdots \infsymb{n}$.
\end{proof}

\section{Hardness}

\begin{proof}[Proof of \cref{generichardness}]
We actually prove a stronger statement, namely that the following problem is hard for $\coNTIME{t}$:
\begin{description}
\item[Given:] A description in $\mM$ of the language $X^{\le m}$ and a
description in $\mN$ of a language $L\subseteq X^{\le m}$, where $X$ is an
alphabet and $m\in\N$.
\item[Question:] Does $\Dclosure{X^{\le m}}\subseteq \Dclosure{L}$ hold?
\end{description}
This is clearly an instance of both $\IncProb{\mM}{\mN}$ and of
$\EqProb{\mM}{\mN}$.  If we show that already this special case is hard, then
so is the case of binary alphabets: Suppose $X=\{a_1,\ldots,a_k\}$ and let
$\gamma\colon X^*\to \{a,b\}$ be the morphism with $\gamma(a_i)=a^ib^{k-i}$ for
$i\in[1,k]$. Then clearly $\Dclosure{\gamma(X^{\le m})}\subseteq
\Dclosure{\gamma(L)}$ if and only if $\Dclosure{X^{\le m}}\subseteq
\Dclosure{L}$. Hence, we only show hardness for the problem above.

Let $K\subseteq Y^*$ belong to $\coNTIME{t}$. Then there is a
$t(n^c)$-time-bounded ($c\ge 1$) Turing machine $M$ with one tape, tape
alphabet $Z\supseteq Y$ (which includes the blank symbol), and state set $Q$
that accepts the complement of $K$.

Our goal is to construct the language $L\subseteq X^{\le m}$ in such a way that
the words in $L$ of length $m$ are precisely those words that do not encode an
accepting computation of $M$. Here, $m$ will be chosen so that if $M$ has an
accepting computation, it is encoded by a word of length $m$.  Then, we will
have $\Dclosure{X^{\le m}}\subseteq \Dclosure{L}$ if and only if $M$ does not
accept the given input word. Our first task is to find a suitable $m$.

Observe that a monotone function $h\colon \N\to\N$ is amplifying if and only if $h(n)\ge
n$ for $n\ge 0$ and there is a $d\ge 1$ such that $h(n^d)\ge h(n)^2$ for large
enough $n$.  Let $g\colon\N\to\N$ be defined as $g(n)=t(n^c)$. Since $t$ is
amplifying, $g$ is as well: for some constant $d$, we have $g(n^d)=t(n^{cd})\ge
t(n^c)^2=g(n)^2$ for large enough $n\in\N$. Since $g$ is amplifying, there is a
constant $e\ge 1$ such that $g(n^e)\ge g(n)\cdot (g(n)+2)$ for all
$n\ge n_0$. We define $f(n)=g(n^e)$.

With these choices, we have: $M$ is time bounded by $g$, the models $\mM$ and
$\mN$ are $\LowerT{g}$ and $\LowerT{f}$, and $f(n)\ge g(n)\cdot(g(n)+2)$ for
$n\ge n_0$.

Now fix $w\in Y^*$ and let $n=|w|$. For the reduction, it means no loss of
generality to assume $n\ge n_0$. We choose $m=f(n)+g(n)+3$. We encode a configuration
of $M$ by a word $uqv$, where $u,v\in Z^*$, $q\in Q$, and $|uv|=g(n)$ (recall
that $M$ is $g$-time-bounded and hence $g$-space-bounded). It means that $M$ is
in state $q$ and its head is at the first position of $v$.  A computation is
then encoded as a word $\#u_1\#\cdots u_k\#u_{k+1}$, where $u_1,\ldots,u_k$
encode the configurations of the computation (in this order) and $u_{k+1}$ is
any suffix in $Z^*$. Since $m=f(n)+g(n)+3\ge g(n)\cdot (g(n)+3)$ and $M$ is
$g$-time-bounded, all computations have encodings where $|\#u_1\cdots
\#u_{k+1}|=m$.

Since $\mM$ and $\mN$ are $\LowerT{g}$ and $\LowerT{f}$, we can construct for
each model finite languages whose longest word has length $g(n)$ or $f(n)$,
respectively. By applying a homomorphism and taking the downward closure, we
can thus construct descriptions of $\Dclosure{\{a^{g(n)}\}}$ and of
$\Dclosure{\{a^{f(n)}\}}$ in each of the models, in polynomial time.
Let $X=Z\cup Q\cup \{\#\}$.
Using rational transductions and simple substitutions, we get
$X^{\le m}=\Dclosure{X^{f(n)+g(n)+3}}$ in $\mM$ and
\[ L_1=\Dclosure{\{a^{i}\#a^{g(n)+1}\#a^{f(n)-i}\mid i\in[0,f(n)]\}} \]
 in $\mN$. Note that $L_1\subseteq \{a,\#\}^{\le m}$.

In the rest of the proof, we construct a (polynomial-sized) rational
transduction $T$ such that $TL_1\subseteq X^{\le m}$ and $(TL_1)\cap X^m$
contains precisely those words that do not encode a computation of $M$ that
accepts $w$. Then, we have clearly shown that the problem described at the
beginning of the proof is hard for $\coNTIME{t}$.

A word $u\in X^*$ of length $m$ can fail to be an accepting computation for $w$
for the following reasons. We decompose $u=u_0\#u_1\#\cdots\#u_{k+1}$.
\begin{enumerate}
\item\label{wrong:firsthash} It does not begin with $\#$, i.e. $u_0\ne\eword$.
\item\label{wrong:tooclose} Two $\#$'s are less than $g(n)+1$ positions apart.
\item\label{wrong:toofar} Two $\#$'s are more than $g(n)+1$ positions apart (without a $\#$ in between).
\item\label{wrong:noconf} Some $u_i$ is not contained in $Z^*QZ^*$.
\item\label{wrong:noinitial} The first configuration $u_1$ is not an initial configuration with input $w$.
\item\label{wrong:notaccepting} The last configuration, i.e. $u_k$, is not accepting.
\item\label{wrong:following} For some $i\in[1,k-1]$, the configuration $u_i$ cannot reach $u_{i+1}$ in one step.
\end{enumerate}
For each of the cases
\labelcref{wrong:firsthash,wrong:tooclose,wrong:toofar,wrong:noconf,wrong:noinitial,wrong:notaccepting,wrong:following},
we shall explain how to obtain a transduction that generates those words from
$L_1$.  If we then have rational transductions $T_1,\ldots,T_7$, we take the
rational transduction $T=T_1\cup\cdots\cup T_7$, which is clearly as desired
above.

Note that the cases
\labelcref{wrong:firsthash,wrong:noconf,wrong:noinitial,wrong:notaccepting} are trivial, so
we consider cases \labelcref{wrong:tooclose,wrong:toofar,wrong:following}. For
\labelcref{wrong:tooclose}, notice that with a constant-sized rational transduction $R_<$, one can obtain
\[ P_< = \Dclosure{\{ a^i \# a^\ell \# a^{g(n)+1-\ell}a^{f(n)-i} \mid i\in[0,f(n)],~ \ell\in[0,g(n)+1] \}}. \]
as $R_<L_1$. Indeed, $R_<$ reads a word from $L_1$ and outputs every letter as read, up to the first $\#$.
Then, before it sees the second $\#$ in the input, it nondeterministically chooses a time to output $\#$ early.
Then, it reads the rest of the input and outputs $a$ for each input letter, be it $a$ or $\#$. Using a similar
strategy, one can obtain
\[ P_> = \Dclosure{\{ a^i \# a^{g(n)+1+\ell} \# a^{f(n)-i-\ell} \mid i\in[0,f(n)],~ \ell\in[0,f(n)-i] \}} \]
using a constant-sized rational transduction $R_>$. Now from $P_<$ and $P_>$, it
is easy to obtain all words of case \labelcref{wrong:tooclose} and
\labelcref{wrong:toofar}, respectively.

The case \labelcref{wrong:following} is also not hard to realize with $L_1$ as
input. We only have to make sure that either the immediate surrounding of the
head is not updated properly or the rest of the tape is not copied correctly.
For words of length $m$ (and those are the only ones where we must produce an
incorrect encoding), the input language $L_1$ gives us, with the two $\#$'s,
two pointers that are precisely $g(n)+1$ positions apart. We can therefore
guarantee that at least one of these errors is present. The details are
very straightforward.
\end{proof}

\begin{proof}[Proof of \cref{expmodels}]
According to \cref{generichardness}, it suffices to show that each
$\mM\in\{\fCFG,\fRBC\}$ is $\LowerT{2^{n}}$.

For $\fCFG$, we can take the well-known grammar with nonterminals
$A_0,\ldots,A_{n}$, start symbol $A_{n}$, and productions $A_i\to
A_{i-1}A_{i-1}$ for $i\in [1,n]$, and $A_0\to a$. It clearly generates
$\{a^{2^n}\}$.

For $\fRBC$, we use a blind $(n+1)$-counter automaton.  We increment
the first counter once and then, for each $i=1,\ldots,n$, we count down counter $i$ and
simultaneously count up counter $i+1$ at twice the speed.  After these $n$
phases, counter $n+1$ contains the value $2^n$. Then, we count down counter
$n+1$ and each time read an $a$. Hence, we accept $\{a^{2^n}\}$.
\end{proof}

\begin{proof}[Proof of \cref{phpitwohardness}]
The \emph{generalized subset sum problem} is the following:
\begin{description}
\item[Given:] Two vectors $u,v\in\N^n$ and $t\in\N$, encoded in binary.
\item[Question:] Is it true that for every $x\in\{0,1\}^n$, there exists a $y\in\{0,1\}^n$ that satisfies
$\langle u,x\rangle + \langle v,y\rangle = t$?
\end{description}
Here, $\langle w,z\rangle$ denotes the scalar product of $w,z\in\Z^n$. This
problem is known to be $\PHPitwo$-complete~\cite{BeKaLaPlRy1997}.

We identify vectors over $\{0,1\}$ of length $n$ with words over $\{0,1\}$ of
length $n$.  Let $u,v\in\N^n$ and $t\in\N$ be an instance of the generalized
subset sum problem and suppose each entry of $u$ and $v$ is encoded with $k$
bits.  Like in \cref{expmodels}, we can easily construct an RBCA $\A$ with $3k$
counters that accepts $\{x\in\{0,1\}^n \mid \exists y\in\{0,1\}^n\colon \langle
u,x\rangle+\langle v,y\rangle=t\}$: As it reads $x$, it uses counters $1,\ldots,k$ to
build up $\langle u,x\rangle$ in counter $k$.  Then, it guesses $y$ bit-by-bit
while using counters $k+1,\ldots,2k$ to build up $\langle v,y\rangle$ in
counter $2k$.  Afterwards, it accumulates $t$ in counter $3k$ using counters
$2k+1,\ldots,3k$. Finally, it counts down counter $3k$ one-by-one and in each
step, decrements counter $k$ or $2k$. In the end, all counters are zero if and
only if $\langle u,x\rangle+\langle v,y\rangle=t$.

Let $\B$ be the obvious $(n+1)$-state NFA that accepts $\{0,1\}^n$. Then we clearly
have $\Dclosure{\langof{\B}}\subseteq\Dclosure{\langof{\A}}$ if and only if our
instance of the generalized subset sum problem is positive.
\end{proof}

\end{document}